\documentclass[MSNbibl,number,citesort,dvips]{arxbj}
\usepackage{graphicx}
\aid{0}
\volume{20}
\issue{1}
\pubyear{2014}
\firstpage{304}
\lastpage{333}
\doi{10.3150/12-BEJ488} 

\makeatletter
\newcommand{\eqref}[1]{(\ref{#1})}
\newtheorem{theorem}{Theorem}[section]
\newtheorem{proposition}[theorem]{Proposition}
\newtheorem{corollary}[theorem]{Corollary}
\newtheorem{lemma}[theorem]{Lemma}
\newremark{remark}[theorem]{Remark}
\newremark{example}[theorem]{Example}
\newproclaim{definition}[theorem]{Definition}
\newcommand{\R}{\mathbb{R}}
\newcommand{\Proba}{\mathbb{P}}
\newcommand{\Cov}{\operatorname{Cov}}
\renewcommand{\P}{\mathbb{P}}
\newcommand{\E}{\mathbb{E}}
\newcommand{\telque}{\dvt }
\newcommand{\FDR}{\operatorname{FDR}}
\newcommand{\FDP}{\operatorname{FDP}}
\newcommand{\mbe}{\mathbb{E}}
\newcommand{\mbr}{\mathbb{R}}
\newcommand{\mbp}{\mathbb{P}}
\newcommand{\mbn}{\mathbb{N}}
\newcommand{\bp}{\mathbf{p}}
\newcommand{\mtc}{\mathcal}
\newcommand{\mbf}{\mathbf}
\newcommand{\wt}[1]{{\widetilde{#1}}}
\newcommand{\wh}[1]{{\widehat{#1}}}
\newcommand{\ol}[1]{\overline{#1}}
\newcommand{\bydef}{:=}
\newcommand{\al}{\alpha}
\newcommand{\eps}{\varepsilon}
\newcommand{\cC}{{\mtc{C}}}
\newcommand{\cH}{{\mtc{H}}}
\newcommand{\cF}{{\mtc{F}}}
\newcommand{\cX}{{\mtc{X}}}
\newcommand{\cP}{{\mtc{P}}}
\newcommand{\cS}{{\mtc{S}}}
\newcommand{\cN}{{\mtc{N}}}
\newcommand{\cY}{{\mtc{Y}}}
\newcommand{\kX}{{\mathfrak{X}}}
\newcommand{\kF}{{\mathfrak{F}}}
\newcommand{\kH}{{\mathfrak{H}}}
\renewcommand{\l}{\ell}
\newcommand{ \cadlag}{c\`adl\`ag }
\makeatother

\begin{document}
\begin{frontmatter}

\title{Testing over a continuum of null hypotheses with False
Discovery Rate control}
\runtitle{Testing over a continuum with FDR control}

\begin{aug}
\author[A]{\fnms{Gilles} \snm{Blanchard}\thanksref{A}\ead[label=e1]{gilles.blanchard@math.uni-potsdam.de}},
\author[B]{\fnms{Sylvain} \snm{Delattre}\thanksref{B}\ead[label=e2]{sylvain.delattre@univ-paris-diderot.fr}}
\and
\author[C]{\fnms{Etienne} \snm{Roquain}\corref{}\thanksref{C}\ead[label=e3]{etienne.roquain@upmc.fr}}
\runauthor{G. Blanchard, S. Delattre and E. Roquain} 
\address[A]{Universit\"at Potsdam, Institut f\"{u}r Mathematik,
Am Neuen Palais 10, 14469 Potsdam, Germany.\\ \printead{e1}}
\address[B]{Universit\'{e} Paris Diderot, LPMA, 4, Place Jussieu, 75252
Paris cedex 05,
France.\\
\printead{e2}}
\address[C]{UPMC Universit\'{e} Paris Diderot, LPMA, 4, Place Jussieu, 75252
Paris cedex 05,
France.\\
\printead{e3}}
\end{aug}

\received{\smonth{11} \syear{2011}}
\revised{\smonth{9} \syear{2012}}

%
\begin{abstract}
We consider statistical hypothesis
testing simultaneously over a fairly general, possibly uncountably infinite,
set of null hypotheses, under the assumption that a suitable single
test (and corresponding $p$-value)
is known for each individual hypothesis. We extend to this setting the
notion of false discovery rate (FDR)
as a measure of type I error.
Our main result studies specific procedures based on the observation of
the $p$-value process.
Control of the FDR at a nominal level is ensured
either under arbitrary dependence of $p$-values, or
under the assumption that the finite dimensional distributions of the
$p$-value process
have positive correlations of a specific type (weak PRDS).
Both cases generalize existing
results established in the finite setting.
Its interest is demonstrated in several non-parametric examples:
testing the mean/signal in a Gaussian white noise model, testing the
intensity of a
Poisson process and testing the c.d.f. of i.i.d. random variables.
\end{abstract}

%
\begin{keyword}
\kwd{continuous testing}
\kwd{false discovery rate}
\kwd{multiple testing}
\kwd{positive correlation}
\kwd{step-up}
\kwd{stochastic process}
\end{keyword}

\end{frontmatter}

\section{Introduction}

\subsection{Motivations}

Multiple testing is a long-established topic in statistics which has
seen a surge of interest
in the past two decades. This renewed popularity is due to a growing
range of applications (such as bioinformatics and
medical imaging) enabled by modern computational possibilities, through which
collecting, manipulating and processing massive amounts of data in very
high dimension has become commonplace.
Multiple testing is in essence a multiple decision problem: each
individual test output is a yes/no (or accept/reject) decision about a
particular question (or null hypothesis)
concerning the generating distribution of some random observed data.

The standard framework for multiple testing is to consider a finite
family of hypotheses and associated tests. However, in many cases
of interest, it is natural to interpret the observed data as the
discretization of an underlying continuously-indexed random process; each
decision (test) is then associated to one of the discretization
points. A first example is that of detecting unusually frequent words
in DNA sequences: a classical model is to consider a Poisson model for
the (non-overlapping) word occurrence process (Robin \cite{Rob2002}), the
observed data being interpreted as a discretized version of this
process. A second example is given in the context of medical imaging,
where the observed pixelized image can be interpreted as a sampled
random process, and the decision to take is, for each pixel, whether
the observed value is due to pure noise or reveals some relevant
activity (pertaining to this setting, see in particular the work of
Perone~Pacifico \textit{et~al.} \cite{PGVW2004,PGVW2007}; see also
Schwartzman, Gavrilov and
Adler \cite{SGA2011}).

Therefore, the present paper explores multiple testing for a (possibly)
uncountably infinite set of hypothesis. With some abuse of language, we
will refer to this
as the \emph{continuous setting} and use loosely the word ``continuous''
in reference to sets in order to mean:
possibly uncountably infinite.
For the above examples, this corresponds to perform testing for the
underlying continuously-indexed process in a direct manner, without
explicit discretization.

\subsection{Contribution and presentation of this work}

The principal contributions of the present work are the following. We
first define a precise, but general in scope,
mathematical setting for multiple testing over a continuous set of hypotheses,
taking particular attention to measurability issues.
Specifically,
we focus on the extension to continuously-indexed
observation (and decision) processes of so-called \emph{step-up}
multiple testing procedures, and the control of
the (continuous analogue of) their \emph{false discovery rate (FDR)}, a
type I error measure which has
gained massive acceptance in the last 15 years for testing in
high-throughput applications.
To this end, we use the tools and analysis developed by (Blanchard and Roquain
\cite{BR2008})
(a programmatic sketch
of the present work can be found in Section~{4.4} of the latter paper).
In particular, we extend suitably to the continuous setting the notion of
positive regressively dependent on a subset (PRDS) condition, which
plays a crucial role in the analysis.
The latter is a general type of dependence condition on the individual
tests' $p$-values allowing to ensure FDR control. An
important difference between the continuous and finite setting is
that the
continuous case precludes the possibility of independent $p$-values,
which is the simplest reference setting considered in the finite case,
so that a more general assumption on dependence structure is necessary
(on this point, see the discussion at the end of Section~\ref{pvalues}).

We have tried as much as possible to make this work self-contained, and
accessible
to readers having little background knowledge in multiple testing.
We begin in the next section with an extended informal discussion of
the framework considered in this paper
in relation to existing literature on non-parametric testing.
Sections~\ref{seccontsetting} and \ref{sectools} of the paper
introduce the necessary notions for multiple testing with
an angle towards stochastic processes, and some specific examples
for the introduced setting.
The main result is exposed in Section~\ref{mainresult}, followed by
its applications to the
examples introduced in Section~\ref{seccontsetting}.
The proof for the main theorem is found in Section~\ref
{secproofmainth}. Extensions and discussions come in Section~\ref
{secdiscuss}, while some technical results are deferred to
\hyperref[secPRDS-appendix]{Appendix} and to the supplementary material
(Blanchard, Delattre and
Roquain \cite{BDR2011-supp}). Throughout the paper, the numbering of the
sections and results of this supplement are preceded by ``S-'' for
clarity (by writing, e.g., Section~S-1).

\subsection{Relationship to non-parametric statistics}
\label{senptest}

Multiple testing over a continuous set of hypotheses has natural ties
with non-parametric statistics.
In the present section, we discuss this link and introduce
informally our goals.
We do not attempt a comprehensive survey of the
very broad field of non-parametric testing, but rather emphasize some
key specificities of the point
of view adopted in the current work.

\textit{Non-parametric testing}.
In order to be more concrete, consider
the classical white noise model $dZ_t = f(t)\,\mathrm{d}t + \sigma\,
\mathrm{d}B_t$, where
$B$ is a Wiener process, $t\in[0,1]$
(this model will be more formally considered as Example~\ref
{exwhitenoise}), with unknown drift function $f\in\cF$,
where $\cF$ is some a priori smoothness class.
The problem of testing various hypotheses about $f$ against
non-parametric alternatives
has received considerable attention since the seminal work of Ingster
\cite
{Ing82,Ing93}. The most common
goal is to test one single qualitative null hypothesis, for instance:
$f$ is identically zero; $f$ is non-negative;
$f$ is monotone; $f$ is convex. For concreteness, consider the first
possibility, denoted
$H^0_*:=\{f \equiv0 \}$.
A common strategy to approach this goal is to consider a collection of
test statistics of the form
$T_{\psi} := \int\psi(t)\,\mathrm{d}Z_t$ for some well chosen
family $\Psi$ of
test functions $\psi$ (normalized so that $\|\psi\|_2=1$).
For each individual test function $\psi$, we have $T_{\psi} \sim\cN
(\int f \psi, \sigma^2)$,
so that this test statistic can be used for a Gauss test of the
``local'' null hypothesis $H^0_{\psi}:=\{f \in\cF; \int f\psi=0\}
\supset H^0_*$.

Intuitively, each statistic $T_\psi$ will have power against a certain
type of alternative.
Taking into account simultaneously all test statistics $T_\psi$ over
$\psi\in\Psi$ now
constitutes a multiple test problem.
The simplest way to combine these is to consider $T_\Psi:= \sup_{\psi
\in\Psi} | T_\psi |$ and reject $H^0_*$
whenever $T_\Psi$ exceeds a certain quantile $\tau$ of its
distribution under the null $H^0_*$.
In multiple testing parlance, this is interpreted as
testing over the hypothesis family $(H^0_{\psi})_{\psi\in\Psi}$
with \emph{weak} control of the
\emph{family-wise error rate} (FWER), which is defined as the
probability to reject falsely one or more
of the considered null hypotheses. Here a local hypothesis $H^{0}_\psi
$ is interpreted as rejected
if $|T_\psi|$ exceeds~$\tau$.
The qualifier ``weak'' refers to the fact that this probability is
controlled at a nominal level only under the global null $H^0_*$ rather
than for all $f\in\cF$.

This connection has been noted and discussed in the literature; for
instance, D{\"u}mbgen and Spokoiny \cite{DS2001},
using the above type construction and multiple testing interpretation,
observe that
``whenever the null hypothesis $H^0_*$ is rejected, we have some
information about where this
violation occurs''. To formalize this idea more precisely, define the
rejection set
$R_\tau:=\{\psi\in\Psi\dvt | T_\psi | > \tau\}$ as
the set of (indices of) hypotheses from the family $(H^0_{\psi})_{\psi
\in\Psi}$ which are deemed false.
In order for $R_\tau$ to be interpretable as intended,
the threshold $\tau$ should now be chosen so that for \emph{any}
$f\in
\cF$,
the probability that the rejection set has non-empty intersection with
(indices of) the set of hypotheses satisfied by $f$, $\{\psi\in \Psi
\dvt f\in H^0_\psi\}$,
is less than a nominal level. This is called \emph{strong} control
of the FWER. This point of view appears to have been only seldom
considered explicitly in non-parametric testing literature;
a recent example is the work of Schmidt-Hieber, Munk and
Duembgen \cite{SMD2011} (in the framework of
density deconvolution),
wherein each individual null hypothesis $H^0_{\psi}$ has a qualitative
interpretation in terms of $f$ being
monotone on some subinterval.

\textit{Type \textup{I} error criteria}.
For adequate (weak or strong) control of the FWER in the example above,
it is clear that
a stochastic control of the deviations of the supremum statistic
$T_\Psi$ is necessary; this, in turn, depends
on the complexity of the set $\Psi$ (typically measured through $L_2$
metric entropy).
As a consequence, FWER-controlling procedures will be more conservative
as the complexity of the family
of the underlying test increases; for instance, in a $d$-dimensional
version of the above example,
or if we simply consider a longer observation interval, the threshold
$\tau$ would have to be more conservative (larger)
to maintain FWER at the same level.

Motivated by the multiple testing point of view,
we consider alternative, less stringent type I error criteria. Let us
still consider the white noise example, in
a slightly modified setup where we are specifically interested in the
family of null hypotheses
$(H^0_{t})_{t \in[0,1]}$ with $H^0_{t}:=\{f \in\cF; f(t) =0\}$. Each
individual null hypothesis $H^0_{t}$ can be tested
using a statistic $T_{\psi_t}$ as defined above;
more precisely, this statistic can test for the null $H^{\delta}_{\psi
_t}:=\{f \in\cF\dvt \int f\psi_t \in[-\delta,\delta]\}$
($\psi_t$ and $\delta$ being chosen adequately so that $H^0_{t}
\subset H^{\delta}_{\psi_t}$ holds).
Define similarly to earlier $R_\tau:=\{t\in[0,1]\dvt| T_{\psi _t} |
> \tau\}$, and introduce
$F_\tau(f) := R_\tau\cap\{t\in[0,1]\dvt f \in H^0_{t}\}$ the set of
incorrectly rejected hypotheses when the true drift is $f$.
To reiterate, FWER control is the requirement that $\mbp_{f}[F_\tau
(f)\neq\varnothing]$ is bounded
at a nominal level for all $f \in\cF$.
Consider now the weaker requirement that the average size of $F_\tau
(f)$ (measured through its Lebesgue measure, denoted $| \cdot |$)
is bounded at some nominal level. We observe via an application of
Fubini's theorem:
\[
\mbe_{f}\bigl[\bigl| F_\tau(f) \bigr|\bigr] = \int
_0^1 \mbp_{f}\bigl[| T_{\psi_t} | >
\tau\bigr]{\mbf{1}\bigl\{f \in H^0_{t}\bigr\}} \,\mathrm{d}t ;
\]
the averaged false reject size is simply the individual test error
level, integrated over the null hypothesis family.
It is the continuous analogue of the so-called \emph{per-comparison
error rate}
(PCER) in multiple testing. To control this quantity, clearly no
multiple testing correction is necessary,
and it is sufficient to choose $\tau$ so that each individual test
has Type I error controlled at the desired level. This criterion, however,
is not very useful in practice: only if the volume the rejection set is
much larger than
the nominal expected volume of errors can we have some trust that the
rejection set contains interesting information.
To address this issue, introduce the average \emph{volume proportion} of
falsely rejected hypotheses:
%
%
\begin{equation}
\label{defFDR}
\FDR(R_\tau,f) := \E_{f} \bigl[
\FDP(R_\tau,f) \bigr]\qquad\mbox{where } \FDP(R_\tau,f) =
\frac{| F_\tau(f) |}{| R_\tau |},
\end{equation}
with the convention $\frac{0}{0}=1$.
The acronyms FDP and FDR stand for \emph{false discovery proportion} and
\emph{rate}, respectively, and the above are
the continuous generalization of corresponding criteria introduced for
finite hypothesis spaces by Benjamini and Hochberg \cite{BH1995}
for a finite number of hypotheses, and which have gained widespread
acceptance since.
Controlling the FDR is a more difficult task than for the PCER,
because of the random denominator $|R_\tau|$ inside the expectation,
and is the main aim of this paper.

As we shall see, a crucial difference of the FDR criterion with respect
to the FWER is that, if a family of tests
of the individual hypotheses is known, there exist relatively generic
procedures called \emph{step-up} to combine
individual tests into a FDR-controlled multiple testing, by finding an
adequate, data-dependent rejection threshold.
In particular, these procedures do not depend on
an intrinsic complexity measure of the initial family, nor of the
control of deviations of suprema of statistics.

To conclude these considerations, it is worth noting that a
FDR-controlled procedure can also
be used for the goal of testing the single ``global null'' hypothesis
$H^0_{*}=\bigcap_{t\in[0,1]} H^0_{t}$, as
in the opening discussion. Namely, if $R_\tau$ is a procedure whose
FDR is controlled at level $\alpha$,
we can reject the global null hypothesis $H^0_{*}$ whenever $| R_\tau
|>0$. Under the global null,
the FDP takes only the values 0 or 1 and precisely coincides with
${\mbf{1}\{| R_\tau |>0\}}$; thus, its expectation
is the probability of type I error for testing the global null this
way, and is bounded by $\alpha$. That this can be
achieved by a generic procedure without explicitly considering the
deviations of a supremum process can seem surprising at first.
Since the focus of this paper is centered on the multiple hypothesis
testing point of view,
we will not
elaborate on this issue further, although a power comparison with the
``standard'' approaches would certainly be of interest.

\textit{
Related work}. The continuous FDR criterion using volume
ratios was introduced before by
Perone~Pacifico \textit{et~al.} \cite{PGVW2004,PGVW2007} to test
non-negativity of the mean of a
Gaussian field. In that work, the authors
follow a two-step approach where the first step consists in a
FWER-controlled multiple testing
based on suprema statistics, as delineated above. This is then used to
define an upper bound on the FDP holding with large probability,
following the principle of so-called \emph{augmentation procedures}
(van~der Laan, Dudoit and
Pollard \cite{vdLDudPol06}) for multiple tests.
An advantage of this approach is that a control of the FDP holding with
large probability is obtained
(which is stronger than a bound on the FDR, its expectation); on the
other hand,
the authors observe that since
the first step is inherently based on FWER control,
it is more conservative than a step-up procedure.
In the present work, we focus on step-up procedures, for which a
probabilistic control of the deviations of suprema statistics is not
needed; this allows us
also to address directly a broader range of applications.

\section{Setting}\label{seccontsetting}

\subsection{Multiple testing: Mathematical framework}\label{secframeworkMT}

Let $X$ be a random variable defined from a measurable space $(\Omega
,\kF)$ to some observation space $(\mathcal{X},\mathfrak{X})$.
We assume that there is a family of probability distributions on
$(\Omega,\kF)$ that induces a subset $\mathcal{P}$ of probability
distributions on $(\mathcal{X},\mathfrak{X})$, which is called the
model. The distribution of $X$ on $(\mathcal{X},\mathfrak{X})$ is
denoted by $P$; for each $P \in\mathcal{P}$ there exists a
distribution on $(\Omega,\kF)$ for which $X\sim P$; it is referred to
as $\P_{X\sim P}$ or simply by $\P$ when unambiguous. The
corresponding expectation operator is denoted $\E_{X\sim P}$ or
$\E$ for short.

We consider a general multiple testing problem for $P$, defined
as follows. Let $\cH$ denote an index space for (null) hypotheses.
To each $h \in\cH$ is associated a known subset $H_h\subset\cP$ of
probability measures on $(\cX,\kX)$.
Multiple hypothesis testing consists in taking a decision, based on
a single realization of the variable $X$, of whether for each $h\in\cH
$ it holds or not that $P\in H_h$ (which is read ``$P$ satisfies
$H_h$'', or ``$H_h$ is true''). We denote by $\cH_0(P)\bydef\{h\in
\cH\telque P \mbox{ satisfies } H_h\}$ the set of true null
hypotheses, and by its complementary $\cH_1(P)\bydef\cH\setminus\cH_0(P)$ the set of false nulls.
These sets are of course unknown because they depend on the unknown
distribution $P$. For short, we will write sometimes $\cH_0$ and $\cH_1$ instead of $\cH_0(P)$ and $\cH_1(P)$, respectively.

As an illustration, if we observe a continuous Gaussian process
$X=(X_h)_{h\in[0,1]^d}$ with a continuous mean function $\mu\dvtx h\in
[0,1]^d\mapsto\mu(t)\bydef\E X_t$, then $P$ is the distribution of
this process, $(\mathcal{X},\mathfrak{X})$ is the Wiener space and
$\mathcal{P}$ is the set of distributions generated by continuous
Gaussian processes having a continuous mean function. Typically, $\cH
=[0,1]^d$ and, for any $h$, we choose $H_h$ equal to the set of
distributions in $\mathcal{P}$ for which the mean function $\mu$
satisfies $\mu(h)\leq0$. This is usually denoted $H_h$: ``$\mu
(h)\leq0$''.
Finally, the set $\cH_0(P)=\{h\in[0,1]^d \telque\mu(h)\leq0\}$
corresponds to the true null hypotheses.
Several other examples are provided below in Section~\ref{secex}.

Next, for a more formal definition of a multiple testing procedure,
we first assume the following:
%
%
\renewcommand{\theequation}{A\arabic{equation}}
\setcounter{equation}{0}
\begin{equation}\label{A1}
\begin{tabular}{p{300pt}@{}}
The index space $\cH$ is endowed with a $
\sigma$-algebra $\mathfrak{H}$ and for all $P \in\cP$,
the set $\cH_0(P)$ of true nulls is assumed to be measurable,
that is, $\cH_0(P) \in\mathfrak{H}$.
\end{tabular}
\end{equation}

%
\begin{definition}[(Multiple testing procedure)]\label{defMTP}
Let $X\dvtx(\Omega,\kF) \rightarrow(\mathcal{X},\mathfrak{X})$ be a
random variable, $\cP$ a model of
distributions of $\cX$, and $\cH$ an index set of null hypotheses.
Assume \textup{\eqref{A1}} holds.
A~multiple testing procedure on $\cH$ is a function $R\dvtx X(\Omega)
\subset\cX\rightarrow\mathfrak{H}$ such that the set
\[
\bigl\{ (\omega,h)\in\Omega\times\cH\dvt h \in R \bigl(X(\omega ) \bigr) \bigr\}
\]
is a $\kF\otimes\kH$-measurable set; or in other terms, that
the process $({\mbf{1}\{h \in R(X)\}})_{h \in\cH}$ is a \emph
{measurable process}.
\end{definition}
The fact that $R$ need only be defined
on the image $X(\Omega)$, rather than on the full space $\cX$, is a
technical detail necessary for later coherence; this introduces no
restriction since $R$ will only be ever applied to possible observed
values of $X$.

A multiple testing procedure $R$ is interpreted as follows:
based on the observation $x=X(\omega)$,
$R(x)$ is the set of null hypotheses that are deemed to be false,
also called set of \emph{rejected} hypotheses.
The set $\cH_0(P)\cap R(x)$ formed of true null hypotheses that are rejected
in error is called the set of \emph{type \textup{I} errors}. Similarly,
the set $\cH_1(P)\cap R^c(x)$ is that of \emph{type \textup{II} errors}.

\subsection{The $p$-value functional and process}
\label{pvalues}

We will consider a very common framework for multiple testing,
where the decision for each null hypothesis $H_h, h\in\cH$, is taken
based on a scalar statistic $p_h(x)\in[0,1]$ called a $p$-value.
The characteristic property of a $p$-value statistic is that if the
generating distribution $P$ is such that the
corresponding null hypothesis is true (i.e., $h \in\cH_0(P)$),
then the random variable $p_h(X)$ should be stochastically lower bounded
by a uniform random variable. Conversely, this statistic is generally
constructed in
such a way that if the null hypothesis $H_h$ is false, its distribution
will be
more concentrated towards the value 0. Therefore, a $p$-value close to
0 is
interpreted as evidence from the data against the validity of the null
hypothesis,
and one will want to reject hypotheses having lower $p$-values.
Informally speaking, based on observation $x$,
the construction of a multiple testing procedure
generally proceeds as follows:
\begin{itemize}[(iii)]
\item[(i)] compute the $p$-value $p_h(x)$ for each individual null
index $h\in\cH$;
\item[(ii)] determine a threshold $t_h(x)$ for each $h\in\cH$,
depending on the whole family $(p_h(x))_{h \in\cH}$;
\item[(iii)] put $R(x) = \{h \in\cH\dvt p_h(x) \leq t_h(x)\}$.
\end{itemize}
To summarize, the rejection set consists of hypotheses whose $p$-values
are lower than a
certain threshold, this threshold being
itself random, depending on the observation $x$ and possibly depending
on $h$.
This will be elaborated in more detail in Section \ref{testproc},
in particular how the threshold function $t_h(x)$ is chosen.
For now, we focus on properly defining the $p$-value functional itself,
the associated process,
and the assumptions we make on them.

Formally, we define the \textit{$p$-value functional} as a mapping $\bp:
\cX\rightarrow[0,1]^\cH$,
or equivalently as a collection of functions $\bp=(p_h(x))_{h\in\cH
}$, each of the functions $p_h \dvtx\cX\rightarrow[0,1]$, $h \in\cH$,
being considered as a scalar statistic that can be computed from the
observed data $x\in\cX$.

We will consider correspondingly the random $p$-values $\omega\in
\Omega\mapsto p_h(X(\omega))$, and $p$-value process $\omega\in
\Omega\mapsto\bp(X(\omega))$.
With some notation overload, we will sometimes drop the dependence on
$X$ and
use the notation $p_h$ and $\bp$ to also denote the \emph{random variables}
$p_h(X)$ and $\bp(X)$ (the meaning -- function of $x$, or random
variable on $\Omega$ --
should be clear from the context).

We shall make the following assumptions on the $p$-value process:
\begin{itemize}
\item Joint measurability over $\Omega\times\cH$:
we assume that the random process $(p_h(X))_{h \in\cH}$ is
a measurable process, that is:
%
%
\begin{equation}
\label{hypmesomega}
\begin{tabular}{p{300pt}@{}}
$(\omega,h)\in(\Omega\times\cH,\kF\otimes
\mathfrak{H} ) \mapsto p_h \bigl(X(\omega) \bigr) \in[0,1]$ \qquad is
(jointly) measurable.
\end{tabular}
\end{equation}

\item For any $P \in\cP$,
the marginal distributions of the $p$-values corresponding to true
nulls are stochastically lower bounded by a uniform random variable on $[0,1]$:
%
%
\begin{equation}
\forall h\in\cH_0(P)\qquad \forall u\in[0,1],\qquad \Proba_{X\sim
P}
\bigl(p_h(X)\leq u \bigr)\leq u.\label{equproppvalues}
\end{equation}
(The distribution of $p_h$ if $h$ lies in $\cH_1(P)$ can be arbitrary).
\end{itemize}

Condition \textup{\eqref{hypmesomega}} is specific to the continuous setting
considered here and will be discussed in more detail in the next section.
Condition \textup{\eqref{equproppvalues}} is the standard characterization
of a single $p$-value statistic in classical (single or multiple)
hypothesis testing.
In general, an arbitrary scalar statistic used to take the rejection
decision on hypothesis $H_h$ can be
monotonically normalized into a $p$-value as follows:
assume $S_h(x)$ is a scalar test statistic, then
\[
p_h(x) = \sup_{P \in H_h} F_{h,P} \bigl(S_h(x)
\bigr)
\]
is a $p$-value in the sense of \textup{\eqref{equproppvalues}},
where $F_{h,P}(t)=\P_{X\sim P}(S_h(X)\geq t)$
(and where the supremum is assumed to maintain the measurability in
$x$, for any fixed $h$).
If the scalar statistic $S_h(x)$ is constructed so that it tends to be
stochastically
larger when hypothesis $H_h$ is false, the corresponding $p$-value
indeed has the desirable property that it is
more concentrated towards 0 in this case. Such test statistics abound in
the (single) testing literature, and a few examples will be given below.

\subsection{Discussion on measurability assumptions}\label{secdiscuss-mes}

Since the focus of the present work is to be able to deal with
uncountable spaces of hypotheses~$\cH$, we have to be somewhat careful
with corresponding measurability assumptions over $\cH$ (a
problem that does not arise when $\cH$ is finite or countable).
The main assumption needed in this regard in order to state properly the
results to come is the joint measurability assumption
appearing in either Definition~\ref{defMTP} (for the multiple testing
procedure) or in \textup{\eqref{hypmesomega}} (for the $p$-value process),
both of which are specific to the uncountable setting.
Essentially, joint measurability will be necessary in order to use
Fubini's theorem on the space $ ( \Omega\times\cH,\mathfrak
{F}\otimes\mathfrak{H} )$, and have the expectation operator
w.r.t. $\omega$ and the integral operator over $\cH$ commute.

If $\cH$ has at most countable cardinality, and is endowed with the
trivial $\sigma$-field comprising all subsets of $\cH$, then \textup{\eqref
{hypmesomega}} is automatically satisfied whenever all
individual $p$-value functions
$p_h\dvtx\cX\rightarrow[0,1]$, $h \in\cH$, are separately measurable,
which is the standard setting in multiple testing.

If $\cH$ is uncountable, a sufficient condition ensuring \textup{\eqref{hypmesomega}}
is the joint measurability of the $p$-value \emph{functional},
%
%
\renewcommand{\theequation}{A\arabic{equation}$^{\prime}$}
\setcounter{equation}{1}
\begin{equation}
\label{hypmesx}
\begin{tabular}{p{300pt}@{}}
$(x,h)\in(\cX\times\cH,\kX\otimes\mathfrak{H} )
\mapsto p_h(x) \in[0,1]$\qquad is (jointly) measurable,
\end{tabular}
\end{equation}
which implies \textup{\eqref{hypmesomega}} by composition. Unfortunately,
\textup{\eqref{hypmesx}} might not always hold.
To see this, consider the following canonical example. Assume the
observation takes the form of a stochastic process indexed by the
hypothesis space itself, $X= \{X_h,h\in\cH\}$. In this case, the
observation space $\cX$ is included in $\mathbb{R}^\cH$.
Furthermore, assume the $p$-value
function $p_h(x)$ is given by a fixed measurable mapping $\psi$ of the
value of $x$ at point $h$, that is, $p_h(x) = \psi(x_h), \forall h
\in\cH$.
In this case
assumption \textup{\eqref{hypmesx}} boils down to the joint measurability of
the evaluation mapping $(x,h) \in\cX\times\cH\mapsto x_h$. Whether
this holds depends on the nature of the space $\cX$. We give some
classical examples in the
next section where the assumption holds;
for example, it is true if $\cX$ is the
Wiener space.\vadjust{\goodbreak}

However, the joint measurability of the evaluation mapping does not
hold if $\cX$
is taken to be the product space $\mbr^\cH$
endowed with the canonical product $\sigma$-field
(indeed, this would imply that any $x\in\mbr^\cH$, i.e.,
any function from $\cH$ into $\mbr$, is measurable).
The more general assumption~\textup{\eqref{hypmesomega}} may still hold, though, but
it generally requires some additional regularity or structural
assumptions on the paths of the process $X$.
In particular, in the above example
if $X= \{X_h,h\in\cH\}$ is a stochastic process having a (jointly)
measurable modification
(and more generally for other examples, if there exists a modification
of $X$
such that \textup{\eqref{hypmesomega}} is satisfied),
we will always assume that we observe such a modification,
so that assumption \textup{\eqref{hypmesomega}} holds.

We have gathered in Section~S-1 of the supplementary
material Blanchard, Delattre and
Roquain \cite{BDR2011-supp} some auxiliary (mostly classical) results
related to
the existence and properties of such modifications.
Lemma~S-1.2 shows that such a (jointly) measurable modification
exists as soon as the process is continuous in probability.
The latter is not an iff condition, but is certainly much weaker
than having continuous paths.

On the other hand, it is important to observe here that a jointly measurable
modification of~$X$, or, for that matter, of the $p$-value process,
might not exist.
Lemma~S-1.1 reproduces a classical argument showing that for
$\cH=[0,1]$, assumption
\textup{\eqref{hypmesomega}} is violated for any modification of a mutually
independent $p$-value process.
Therefore, for an uncountable space of hypotheses~$\cH$, assumption
\textup{\eqref{hypmesomega}} precludes
the possibility that the $p$-values $\{p_h,h\in\cH\}$ are mutually
independent.
This contrasts strongly with the situation of
a finite hypothesis set~$\cH$, where mutual independence of the
$p$-values is
generally considered the reference case.

A final issue is to which extent the results exposed in the remainder
of this work
depend on the (jointly) measurable modification chosen for the
underlying stochastic process.
Lemma~S-1.4 elucidates this issue by showing that this
is not the case,
because the FDR (the main measure of type I error, which will be
formally defined in Section~\ref{seFDRdef})
is identical for two such modifications.

\subsection{Examples}\label{secex}

To illustrate the above generic setting, let us consider the following examples.

%
\begin{example}[(Testing the mean of a process)]\label{exmeanprocess}
Assume that we observe the realization of a real-valued process
$X=(X_t)_{t\in[0,1]^d}$ with an unknown (measurable) mean function
$\mu\dvtx t\in[0,1]^d\mapsto\mu(t)\bydef\E X_t$.
We take $\cH=[0,1]^d$ and want to test simultaneously for each $t\in
[0,1]^d$ the null hypothesis $H_t$: ``$\mu(t)\leq0$''.
Assume that for each $t$ the marginal distribution of $(X_t-\mu(t))$
is known,
does not depend on $t$ and has upper-tail function $G$ (e.g.,
$X$ is a Gaussian process with marginals $X_t\sim\mathcal{N}(\mu
(t),1)$). We correspondingly define the $p$-value process $\forall t
\in[0,1]^d$, $p_t(X)=G(X_t)$, which satisfies \textup{\eqref
{equproppvalues}}. Next, the measurability assumption \textup{\eqref
{hypmesomega}} follows from a regularity assumption on $X$:
\begin{itemize}
\item if we assume that the process $X$ has continuous
paths, $X \dvtx\omega\mapsto(X_t(\omega))_t$ can be seen as taking
values in the Wiener space $\cX=\mathcal
{C}_{[0,1]^d}=C([0,1]^d,\mathbb{R})$ of continuous functions from
$[0,1]^d$ to $\mathbb{R}$. (In this case,
the Borel $\sigma$-field corresponding to the supremum norm topology
on $\cC_{[0,1]^d}$ is the trace of
the product $\sigma$-field on $\mathcal{C}_{[0,1]^d}$, and $X$ is
measurable iff all its coordinate projections are.) Furthermore,
the $p$-value function can be written as
\[
(x,t)\in\cC_{[0,1]^d}\times[0,1]^d \mapsto p_t(x)=G
\bigl(x(t) \bigr) \in[0,1].
\]
The evaluation functional $(x,t)\in\cC_{[0,1]^d} \times[0,1]^d
\mapsto x(t)$ is jointly measurable because it is continuous, thus
$p_t(x)$ is jointly measurable by composition and \textup{\eqref{hypmesx}}
holds, hence also~\textup{\eqref{hypmesomega}};
\item if $d=1$ and the process $X$ is \cadlag, the
random variable $X$ can be seen as taking values in the Skorohod space
$\cX=\mathcal{D}\bydef D([0,1],\mathbb{R})$ of \cadlag functions
from $[0,1]$ to $\mathbb{R}$. In this case, the Borel $\sigma$-field
generated by the Skorohod topology is also the trace of
the product $\sigma$-field on $\mathcal{D}$ (see, e.g., Theorem~14.5
page~121 of Billingsley \cite{Bill1999}). Moreover, the evaluation functional
$(x,t)\mapsto x(t)$ is jointly measurable, as for any
\cadlag funtion~$x$, it is the pointwise limit of the jointly
measurable functions $\zeta_n$:
$
(x,t)\mapsto\zeta_n(x,t) \bydef\sum_{k=1}^{2^n} x(k2^{-n}) {\mbf
{1}\{(k-1)2^{-n}\leq t < k2^{-n}\}}+x(1){\mbf{1}\{t=1\}},
$
therefore \textup{\eqref{hypmesx}} is fulfilled by composition, hence
also~\textup{\eqref{hypmesomega}};

\item
assume that $X$ is a Gaussian
process defined on the space $\cX=\mbr^{[0,1]^d}$ endowed with the
canonical product $\sigma$-field, and with a covariance function
$\Sigma(t,t')$ such that $\Sigma$ is continuous on all points
$(t,t)$ of the diagonal and takes a constant (known) value $\sigma^2$
on those points.

This assumption is not sufficient to ensure that $X$ has a continuous
version, but it ensures that $(X_t)$ is continuous in $L^2$ and hence
in probability; Lemma~S-1.2 then states that $X$ has a
modification such that the evaluation functional is jointly measurable.
Assuming that such a jointly measurable modification is observed,
we deduce that \textup{\eqref{hypmesomega}} holds for the associated
$p$-value process.
\end{itemize}
\end{example}

%
\begin{example}[(Testing the signal in a Gaussian white noise
model)]\label{exwhitenoise}
Let us consider the Gaussian white noise model
$dZ_t = f(t)\,\mathrm{d}t + \sigma\,\mathrm{d}B_t$, $t\in[0,1]$,
where $B$ is a Wiener
process on $[0,1]$ and $f\in C([0,1])$ is a continuous signal function.
For simplicity, the standard deviation $\sigma$ is assumed to be equal
to $1$.
Equivalently, we assume that we can observe the stochastic integral of
$Z_t$ against any test function in $L^2([0,1])$, that is, that we
observe the Gaussian process $(X_g)_{g\in L^2([0,1])}$ defined by
\[
X_g \bydef\int_0^1 g(t) f(t) \,
\mathrm{d}t + \int_0^1 g(t) \,
\mathrm{d}B_t,\qquad g\in L^2 \bigl([0,1] \bigr).
\]
Formally, the observation space is the whole space $\cX=\R^{L^2([0,1])}$, endowed with the product $\sigma$-field.
However, in the sequel, we will use the observation of the process $X$
only against a ``small'' subspace of functions of $L^2([0,1])$.

Let us consider $\cH=[0,1]$ and the problem of testing for each $t\in
\cH$, the null $H_t$: ``$f(t)\leq0$'' (signal non-positive).
We can build $p$-values based upon a kernel estimator in the following way.
Consider a kernel function $K\in L^2( \R)$, assumed positive on
$[-1,1]$ and zero elsewhere, and denote by $K_t \in L^2([0,1])$ the
function $K_t(s) \bydef K((t-s)/\eta)$, where $0<\eta\leq1$ is a
bandwidth to be chosen.\vadjust{\goodbreak}
Let us consider the process $\wt{X}_t:=X_{K_{t}}, t\in[0,1]$. From
Lemma~S-1.3, $\wt{X}$ has a modification which is
jointly measurable in $(\omega,t)$. Clearly, this implies that there
exists a modification of the original process $X$ such that $\wt{X}$
is jointly measurable in $(\omega,t)$, and we assume that we observe
such a modification.
For any $t\in[0,1]$, letting $c_{K,t}\bydef\int_0^1 K((t-s)/\eta)
\,\mathrm{d}s>0$ and $v_{K,t}\bydef\int_0^1 K^2((t-s)/\eta) \,
\mathrm{d}s\geq
c_{K,t}^2>0$, we can consider the following standard estimate of $f(t)$:
%
%
\renewcommand{\theequation}{\arabic{equation}}
\setcounter{equation}{1}
\begin{eqnarray}\label{equestimatefdet}
\wh{f}_\eta(t)&\bydef& c_{K,t}^{-1}
X_{K_{t}}
\nonumber
\\[-8pt]
\\[-8pt]
\nonumber
&= &c_{K,t}^{-1} \int_0^1 K
\biggl(\frac{t-s}{\eta}\biggr) f(s)\,\mathrm{d}s + c_{K,t}^{-1}
\int_0^1 K\biggl(\frac{t-s}{\eta}\biggr) \,
\mathrm{d}B_s.
\end{eqnarray}
Assume that there is a known $\delta_{t,\eta}>0$ such that for any
$t$ with $f(t)\leq0$, we have the upper-bound
%
%
\begin{equation}
\label{eqregcondf} \E\wh{f}_\eta(t) = c_{K,t}^{-1}
\int_0^1 K\biggl(\frac{t-s}{\eta }\biggr) f(s)
\,\mathrm{d}s \leq\delta_{t,\eta}.
\end{equation}
For instance, this holds if we can assume a priori knowledge on the regularity
of $f$, of the form $\sup_{s:|s-t|\leq\eta} |f(s)-f(t)| \leq\delta_{t,\eta}$.
Then, the statistics $(\wh{f}_\eta(t))_t$ can be transformed into a
$p$-value process
in the following way:
%
%
\begin{equation}
\label{equpvalueBBG} p_t(X)=\overline{\Phi} \biggl(\frac{\wh
{f}_\eta(t)-\delta_{t,\eta
}}{ v_{K,t}^{1/2}/c_{K,t}}
\biggr),
\end{equation}
where $\overline{\Phi}(w)\bydef\P(W\geq w)$, $W\sim\mathcal
{N}(0,1)$, is the upper tail distribution of a standard Gaussian distribution.
The $p$-value process \eqref{equpvalueBBG} satisfies \textup{\eqref
{equproppvalues}}, because
for any $t$ with $f(t)\leq0$ and any $u\in[0,1]$,
\begin{eqnarray*}
\P \bigl( p_t(X)\leq u \bigr) &=& \P \bigl( \wh{f}_\eta(t)-
\delta_{t,\eta} \geq v_{K,t}^{1/2}/c_{K,t} \ol{
\Phi}^{-1}(u) \bigr)
\\
&\leq&\P \bigl( c_{K,t} \bigl(\wh{f}_\eta(t)-\E
\wh{f}_\eta(t) \bigr) /v_{K,t}^{1/2} \geq\ol{
\Phi}^{-1}(u) \bigr)
\\
&=& u,
\end{eqnarray*}
because $\int_0^1 K_t(s) \,\mathrm{d}B_s \sim\mathcal{N}(0, v_{K,t})$.
Moreover, the $p$-value process \eqref{equpvalueBBG} satisfies \textup{\eqref
{hypmesomega}}, since we assumed $(X_{K_{t}})_t\in[0,1]$ to be
jointly measurable in $(\omega,t)$.
\end{example}

%
\begin{example}[(Testing the c.d.f.)]\label{excdf}
Let $X=(X_1,\ldots,X_m)\in\cX=\mathbb{R}^m$ be a $m$-uple of i.i.d.
real random variables of common continuous c.d.f. $F$. For $\cH= I$ an
interval of $\R$ and a given benchmark c.d.f. $F_0$, we aim to test
simultaneously for all $t\in I $ the null $H_ t\dvt$ ``$F(t)\leq F_0(t)$''.
The individual hypothesis $H_ t$ may be tested using the $p$-value
%
%
\renewcommand{\theequation}{\arabic{equation}}
\setcounter{equation}{4}
\begin{equation}
\label{pvaluefit} p_t(X)= G_t \bigl(m
\mathbb{F}_m(X,t) \bigr),
\end{equation}
where
$\mathbb{F}_m(X,t)=m^{-1} \sum_{i=1}^m {\mbf{1}\{X_i\leq t\}}$ is the
empirical c.d.f. of $X_1,\ldots,X_m$ and where $G_t(k)=\mbp[Z_t\geq k]$, $Z_t\sim\mathcal{B}(m,F_0(t))$, is the upper-tail
function of a
binomial distribution of parameter $(m,F_0(t))$.
The conditions \textup{\eqref{hypmesomega}} and \textup{\eqref{equproppvalues}} are
both clearly satisfied.

Figure~\ref{figtestcdf} provides a realization of the $p$-value
process \eqref{pvaluefit}~when testing for all $t\in[0,1] $ the null
$H_ t\dvt$ ``$F(t)\leq t$'' when $F$ comes from a mixture of beta
distributions. The correct/erroneous rejections are also pictured for
the simple procedure $R(X)=\{t\in[0,1]\telque p_t(X)\leq0.4\}$.
\end{example}

%
\begin{example}[(Testing the intensity of a Poisson process)]\label{exPoisson}
Assume we observe $(N_t)_{t\in[0,1]}\in\cX=D([0,1],\mathbb{R})$ a
Poisson process with intensity $\lambda\dvtx [0,1] \rightarrow\R^+ \in
L^1(d\Lambda)$, where $\Lambda$ denotes the Lebesgue measure on
$[0,1]$. For each $t\in[0,1]$, we aim to test $H_ t$: ``$\lambda
(t)\leq\lambda_0(t)$''
where $\lambda_0(\cdot)>0$ is a given benchmark intensity.
Assume that for a given bandwidth $\eta\in(0,1]$, there is a known
upper bound $\delta_{t,\eta}$ for $\int_{(t-\eta)\vee0}^{(t+\eta
)\wedge1} \lambda(s)\,\mathrm{d}s$ that holds true for any $t$ such that
$\lambda(t)\leq\lambda_0(t)$.

%
\begin{figure}

\includegraphics{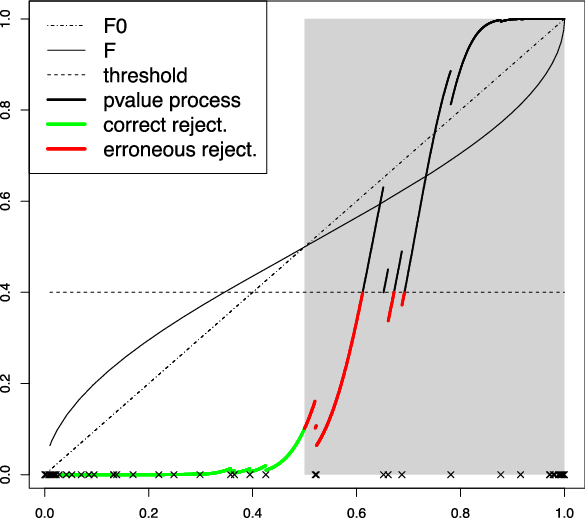}

\caption{Plot of a realization of the $p$-value process as defined in
\protect\eqref{pvaluefit}
for the c.d.f. testing, together with $F_0$ and $F$, for $F_0(t)=t$ and
$F(t)=0.5 F_1(t)+0.5 F_2(t)$, where $F_1$ (resp., $F_2$) is the c.d.f.
of a beta distribution of parameter $(0.5, 1.5)$ (resp., $(1.5,0.5)$).
The region where the null hypothesis ``$F(t)\leq F_0(t)$'' is true is
depicted in grey color.
The crosses correspond to the elements of $\{X_{i},1\leq i\leq m\}$; $m=50$.
The correct/erroneous rejections refer to the procedure $R(X)=\{t\in
[0,1]\telque p_t(X)\leq0.4\}$ using the threshold~$0.4$.}\vspace*{2pt}
\label{figtestcdf}
\end{figure}

For instance, we can choose $\delta_{t,\eta}=((t+\eta)\wedge
1-(t-\eta)\vee0)(\lambda_0(t) + \sup_{s:|t-s|\leq\eta}|\lambda
(t)-\lambda(s)|)$
(assuming knowledge on the regularity of $\lambda$ is available a priori).
For any $t\in[0,1]$, the variable $N_{(t+\eta)\wedge1}-N_{(t-\eta
)\vee0}$ follows a Poisson variable of parameter $\int_{(t-\eta)\vee
0}^{(t+\eta)\wedge1} \lambda(s)\,\mathrm{d}s$. Since the latter
parameter is
smaller than $\delta_{t,\eta}$
as soon as $\lambda(t)\leq\lambda_0(t)$, the following $p$-value
process satisfies~\textup{\eqref{equproppvalues}}:
%
%
\renewcommand{\theequation}{\arabic{equation}}
\setcounter{equation}{5}
\begin{equation}
\label{pvaluespois} p_t(X)=G_t (N_{(t+\eta)\wedge1}-N_{(t-\eta)\vee0}
),
\end{equation}
where for any $k\in\mathbb{N}$, $G_t(k)$ denotes $\mbp[Z\geq k]$
for $Z$ a Poisson distribution of parameter
$\delta_{t,\eta}$.
Moreover, the $p$-value process fulfills condition \textup{\eqref{hypmesx}},
because $(N_t)$ is a \cadlag process, so that arguments similar to
those of Example \ref{exmeanprocess} apply. Thus, \textup{\eqref
{hypmesomega}} also holds.
\end{example}

\section{Main concepts and tools}\label{sectools}

\subsection{False discovery rate}
\label{seFDRdef}

Following the usual philosophy of hypothesis testing,
one wants to ensure some control over type I errors committed by the procedure.
As discussed in Section~\ref{senptest}, in the present work we focus
on a generalization to a continuum
of hypotheses of the \emph{False Discovery Rate} (FDR).
For a finite number of null hypotheses, the FDR, as introduced by
Benjamini and Hochberg \cite
{BH1995} (see also Seeger \cite{See1968}),
is defined as the average proportion
of type I errors in the set of all rejected hypotheses. To extend this
definition to a possibly uncountable space,
following Perone~Pacifico \textit{et~al.} \cite{PGVW2004,PGVW2007},
we quantify this proportion by a volume ratio, defined with respect to
a finite measure $\Lambda$ on $(\cH, \mathfrak{H})$ (the usual
definition over a finite space is recovered by taking $\Lambda$ equal
to the counting measure).

%
\begin{definition}[(False discovery proportion, false discovery rate)]
Let $\Lambda$ be a finite positive measure on $(\cH,\kH)$.
Let $R$ be a multiple testing procedure on $\cH$. The false
discovery rate (FDR) of~$R$ is defined as the average of the false
discovery proportion (FDP):
%
%
\renewcommand{\theequation}{\arabic{equation}}
\setcounter{equation}{6}
\begin{equation}
\label{defFDP} \forall P \in\cP, \forall x \in X(\Omega),\qquad \FDP \bigl(R(x),P
\bigr)\bydef\frac{\Lambda(R(x)\cap\cH_0(P) )}{\Lambda
(R(x) )} {\mbf{1}\bigl\{\Lambda \bigl(R(x) \bigr)>0\bigr\}}
\end{equation}
and
%
%
\begin{equation}
\forall P \in\cP,\qquad \FDR(R,P)\bydef\mbe_{X \sim
P}\bigl[\FDP
\bigl(R(X),P \bigr)\bigr].
\end{equation}
\end{definition}
The indicator function in \eqref{defFDP} means that the
ratio is taken equal to zero whenever the denominator is zero.
Observe that, due to the joint measurability assumption in
Definition~\ref{defMTP} of a multiple testing procedure,
both of the above quantities are well-defined (the FDP is
only formally defined over the image of $\Omega$ through $X$
since only on this set is the measurability of
$R(x)$ guaranteed by the definition.
In particular, it is defined for $P$-almost all $x\in\cX$).

As illustration, in the particular realization of the $p$-value process
pictured in Figure~\ref{figtestcdf}, if we denote by ``Red'' (resp.,
``Green'') the length of the interval corresponding to the projection of
the red (resp., green) part of the $p$-value process on the $X$-axis,
the FDP of the procedure $R(X)=\{t\in[0,1]\telque p_t(X)\leq0.4\}$ is
$\mathrm{Red}/(\mathrm{Red} + \mathrm{Green})$. A similar interpretation for the FDP holds in
Figure~\ref{figtestpois}.\vadjust{\goodbreak}

Finding a procedure $R$ with a $\FDR$ smaller than or equal to $\alpha
$ has the following interpretation: on average, the volume proportion
of type I errors among the rejected hypotheses is smaller than $\alpha
$. This means that the procedure is allowed to reject in error some
true nulls but in a small (average) proportion among the rejections.
For a pre-specified level $\alpha$, the goal is then to determine
multiple testing procedures $R$ such that for any $P \in\cP$, it
holds that $\FDR(R,P)\leq\alpha$.
(In fact, the statement need only hold for $P \in\cP\cap\bigcup_{h\in\cH}H_h$, since outside of this set $\cH_0(P)=\varnothing$ and
the FDR is 0.) The rest of the paper will concentrate on establishing
sufficient conditions under which the FDR is controlled at a fixed
level $\alpha$. Under this constraint, in order to get a procedure
with good power properties (i.e., low type II error), it is,
generally speaking, desirable that $R$ rejects as many nulls as
possible, that is, has volume $\Lambda(R)$ as large as
possible.

\subsection{Step-up procedures}\label{secstep-up} \label{testproc}

In what follows, we will focus on a particular form of multiple testing
procedures which can be written as function of the $p$-value family
$\mathbf{p}(x)=(p_h(x))_{h\in\cH}$.

First, we define a parametrized family of possible rejection sets
having the following form: for a given \emph{threshold function}
$\Delta
\dvt(h,r) \in\cH\times\mathbb{R}^+\mapsto\mathbb{R}^+$, we define
for any $r\geq0$ the sub-level set
%
%
\renewcommand{\theequation}{\arabic{equation}}
\setcounter{equation}{8}
\begin{equation}
\label{equlevelset} \forall x \in\cX,\qquad L_\Delta(x,r) \bydef \bigl\{ h\in
\cH \telque p_h(x) \leq\Delta(h,r) \bigr\} \subset\cH.
\end{equation}
For short, we sometimes write $L_\Delta(r)$ instead of $L_\Delta
(x,r)$ when unambiguous.
We will more particularly focus on threshold functions $\Delta$ of the
product form $\Delta(h,r)=\alpha\pi(h) \beta(r)$, where $\alpha\in
(0,1)$ is a positive scalar (\emph{level}), $\pi\dvtx \cH\rightarrow
\mathbb{R}^+$ is measurable (\emph{weight function}), and $\beta\dvtx
\mathbb{R}^+ \rightarrow\mathbb{R}^+$ is non-decreasing and
right-continuous (\emph{shape function}).
Clearly, this decomposition is not unique, but will be practical for
the formulation of the main result.

Given a threshold function $\Delta$ of the above form, we will be
interested in a particular, data-dependent choice of the parameter $r$
determining the rejection set, called \emph{step-up procedure}.

%
\begin{definition}[(Step-up procedure)]\label{defcontstepup}
Let $\Delta(h,r)=\alpha\pi(h) \beta(r)$ a threshold function with
$\alpha\in(0,1)$; $\pi\dvtx\cH\rightarrow\mathbb{R}^+$ measurable and
$\beta\dvtx\mathbb{R}^+\rightarrow\mathbb{R}^+$ non-decreasing and
right-continuous.
Then the \textit{step-up multiple testing procedure} $R$ on $(\cH
,\Lambda)$ associated to
$\Delta$,
is defined by
%
%
\begin{equation}\label{equSU} \hspace*{-16pt}
\forall x \in X(\Omega),\qquad R(x) = L_\Delta \bigl(x,\wh{r}(x)
\bigr) \quad\mbox{where } \wh{r}(x)\bydef\max \bigl\{r \geq0 \telque \Lambda
\bigl(L_\Delta(x,r) \bigr)\geq r \bigr\} .
\end{equation}
\end{definition}

Note that $\wh{r}$ above is well-defined:
first, since $x \in X(\Omega)$ and from assumption \textup{\eqref
{hypmesomega}}, the function $h\mapsto p_h(x)-\alpha\pi(h) \beta
(r)$ is measurable;
thus $L_\Delta(x,r)$ is a measurable set of $\cH$, which in turn
implies that $\Lambda(L_\Delta(x,r))$ is well-defined.
Secondly, the supremum of $\{r \geq0 \telque\Lambda(L_\Delta
(x,r))\geq r\}$ exists because $r=0$ belongs to this set and $M=\Lambda
(\cH)$ is an upper bound. Third, this supremum is a maximum because
the function $r\mapsto\Lambda(L_\Delta(x,r))$ is non-decreasing
(right-continuity is not needed for this).

We should ensure in Definition~\ref{defcontstepup} that a step-up procedure
satisfies the measurability requirements of Definition~\ref{defMTP}.
This is proved separately in Section~\ref{secstep-upproof}.
In that section, we also check that the equality $\Lambda(L_\Delta
(x,\wh{r}(x)))=\wh{r}(x)$ always holds. Hence, $\wh{r}(x)$ is the
largest intersection point between the function $r\mapsto\Lambda
(L_\Delta(x,r))$
giving the volume of the candidate rejection sets as a function of $r$,
and the identity line $r\mapsto r$.

To give some basic intuition behind the principle of a step-up
procedure, consider for simplicity
that $\pi$ is a constant function, so that the family defined by
\eqref{equlevelset} are
ordinary sub-level sets of the $p$-value family. The goal is to find a
suitable common
rejection threshold $t$ giving rise to rejection set $R_t$. Assume also
without loss of generality
that $\Lambda(\cH)=1$.
Now consider the following heuristic. If the threshold $t$
is deterministic, any $p$-value associated to a true null hypothesis,
being stochastically lower bounded by
a uniform variable, has probability less than $t$ of being rejected in
error. Thus, we
expect on average a volume $t \Lambda(\cH_0) \leq t$ of erroneously rejected
null hypotheses. If we therefore use $t$ as a rough upper bound of the
numerator in the definition
\eqref{defFDP} of the FDP or FDR, and we want the latter to be less
than $\alpha$, we obtain the
constraint $t/\Lambda(R_t) \leq\alpha$, or equivalently $\Lambda
(R_t) \geq\alpha^{-1}t$.
Choosing the largest $t$ satisfying this heuristic constraint is
equivalent to the step-up
procedure wherein $\beta(u)=u$. The choice of a different shape
function with $\beta(u)\leq u$ can be
interpreted roughly as a pessimistic discount
to compensate for various inaccuracies in the above heuristic argument
(in particular the fact
that the obtained threshold is really a random quantity).

In the case where $\cH$ is finite and $\Lambda$ is the counting
measure, it can be seen that the above definition recovers the usual
notion of
step-up procedures (see, e.g., Blanchard and Roquain~\cite{BR2008});
in particular, the
linear shape function
$\beta(u)=u$ gives rise to the celebrated linear step-up procedure of
Benjamini and Hochberg \cite{BH1995}.

\subsection{PRDS conditions}\label{secPRDS}

To ensure control of the FDR criterion, an important role is played by
structural assumptions on the dependence of the $p$-values. While
the case of independent $p$-values is considered as the reference
setting in the case where $\cH$ is finite, we recall that for an
uncountable set $\cH$, we cannot assume mutual independence of the
$p$-values since this would contradict our measurability assumptions
(see concluding discussion of Section~\ref{secdiscuss-mes}).

We will consider two different situations in our main result: first, if
the dependence of the $p$-values can be totally arbitrary, and
secondly, if a form of positive dependence is assumed. This is the
latter condition which we define more precisely now.
We consider a generalization to the case of infinite, possibly
uncountable space $\cH$, of the notion of positive regression
dependence on each one from a subset (PRDS) introduced by Benjamini and
Yekutieli \cite
{BY2001} in the case of a finite set of hypotheses.

For any finite set $\mathcal{I}$, a subset $D\subset[0,1]^\mathcal
{I}$ is called
\textit{non-decreasing} if for all $\mbf{z},\mbf{z}' \in
[0,1]^\mathcal{I}$
such that $\mbf{z} \leq\mbf{z}'$ (i.e., $\forall h\in\mathcal{I},
z_h\leq
z'_h$), we have $\mbf{z} \in D \Rightarrow\mbf{z}' \in D$.

%
\begin{definition}[(PRDS conditions for a finite
$p$-value family)]\label{defPRDSfinite}
Assume $\cH$ to be finite.
For $\cH'$ a
subset of $\cH$, the $p$-value family
$\mathbf{p}(X)=(p_h(X))_{h\in\cH}$ is said to be \textit{weak PRDS
on $\cH'$} for the distribution $P$,\vadjust{\goodbreak}
if for any $h\in\cH'$, for any measurable non-decreasing set $D$ in
$[0,1]^\cH$, the function
$u\in[0,1]\mapsto\Proba(\mbf{p}(X)\in D | p_h(X) \leq u)$
is non-decreasing on $\{u\in[0,1]\telque\P(p_h(X)\leq u)>0\}$; it is
said to be \textit{strong PRDS}
if the function $ u \mapsto\Proba(\mbf{p}(X)\in D | p_h(X)=u)$ is
non-decreasing.
\end{definition}
To be completely rigorous, observe that the conditional probability
with respect to the event $\{p_h(X) \leq u\}$ is defined pointwise
unequivocally whenever this event has positive probability, using
a ratio of probabilities; while the conditional probability
with respect to $p_h(X) = u$ can only be defined via conditional
expectation, and is therefore only defined up to a $p_h(X)$-negligible
set. Hence, in the definition of strong PRDS,
strictly speaking, we only require that the conditional probability
coincides $p_h(X)$-a.s. with a non-decreasing function.

%
\begin{definition}[(Finite dimensional PRDS
conditions for a $p$-value process)]\label{deffinitePRDS}
For $\cH'$ a subset of~$\cH$, the $p$-value process
$\mathbf{p}(X)=(p_h(X))_{h\in\cH}$ is said to be \textit{finite
dimensional weak PRDS on~$\cH'$} (resp., \textit{finite dimensional
strong PRDS on~$\cH'$}) for the distribution $P$,
if for any finite subset $\cS$ of $\cH$,
the finite $p$-value family $\bp_{\cS}(X)=(p_h(X))_{ h \in\cH\cap
\cS}$ is
weak PRDS on $\cH' \cap\cS$ (resp., strong PRDS on $\cH' \cap\cS$)
for the distribution $P$.
\end{definition}

While the finite dimensional weak PRDS property will be sufficient to
state our main result, the strong PRDS property is sometimes easier to handle.
Hence, it is important to note that the finite dimensional strong PRDS
property implies the finite dimensional weak PRDS property, as we
establish in Lemma~S-2.2, by using a standard
argument pertaining to classical multiple testing theory for a finite
set of hypotheses.

Finally, Benjamini and Yekutieli \cite{BY2001} (Section~3.1 therein)
proved that the $p$-value
family corresponding to a finite Gaussian random vector are (finite)
strong PRDS as soon as all the coefficient of the covariance matrix are
non-negative. This equivalently proves the following result.
%
%
\begin{lemma}\label{lem-PRDS-Gauss}
Let $\mathbf{p}(X)=(p_h(X))_{h\in\cH}$ be a $p$-value process of the
form $p_h(X)=G(X_h)$, $h\in\cH$, where $X=(X_h)_{h\in\cH}$ is a
Gaussian process and where $G$ is continuous decreasing from $\R$ to $[0,1]$.
Assume that the covariance function $\Sigma$ of $X$ satisfies
%
%
\begin{equation}
\label{equsigmapos} \forall h,h'\in\cH,\qquad \Sigma \bigl(h,h'
\bigr)\geq0.
\end{equation}
Then the $p$-value process is finite dimensional strong PRDS (on any subset).
\end{lemma}

\section{Control of the FDR}\label{mainresult}

In this section, our main result is stated and then illustrated with
several examples.

\subsection{Main result}

The following theorem establishes our main result on sufficient
conditions to
ensure FDR control at a specified level for step-up procedures.
It is proved in Section~\ref{secproofmainth}.
%
%
\begin{theorem}
\label{mainthm}
Assume that the hypothesis space $\cH$ satisfies \textup{\eqref{A1}} and is
endowed with a finite measure $\Lambda$. Let $\bp(X)=(p_h(X))_{h\in
\cH}$ be a $p$-value process satisfying the conditions \textup{\eqref
{hypmesomega}} and \textup{\eqref{equproppvalues}}. Denote $R$ the step-up
procedure on $(\cH,\Lambda)$
associated to a threshold function of the product form $\Delta
(h,r)=\alpha\pi(h)\beta(r)$, with $\alpha\in(0,1)$, $\beta$ a
non-decreasing right-continuous shape function and $\pi$ a probability
density function on $\cH$ with respect to $\Lambda$.
Then for any $P\in\cP$, letting $\Pi(\cH_0(P)) \bydef\int_{h \in
\cH_0(P)} \pi(h) \,\mathrm{d}\Lambda(h)$, the inequality
%
%
\begin{equation}
\FDR(R,P) \leq\al\Pi \bigl(\cH_0(P) \bigr) \qquad (\leq\al)
\label{equFDRcontrol}
\end{equation}
holds in either of the two following cases:
\begin{enumerate}
\item$\beta(x)=x$ and the $p$-value process $\mathbf{p}$ is finite
dimensional weak PRDS on $\cH_0(P)$ for the distribution $P$;
\item the function $\beta$ is of the form
%
%
\begin{equation}
\label{condbeta} \beta_\nu(x) = \int_0^x
u \,\mathrm{d}\nu(u) ,
\end{equation}
where $\nu$ is an arbitrary probability distribution on $(0,\infty)$.
\end{enumerate}
\end{theorem}

Since $\pi$ is taken as a probability density function on $\cH$ with
respect to $\Lambda$, the FDR in \eqref{equFDRcontrol} is upper
bounded by $\alpha\Pi(\cH_0)\leq\alpha\Pi(\cH) = \alpha$, so
that the corresponding step-up procedure provides FDR control at level~$\alpha$.
As an illustration, a typical choice for $\pi$ is the constant
probability density function $\forall h\in\cH, \pi(h)=1/\Lambda(\cH
)=M^{-1}$.

According to the standard philosophy of (multiple) testing, while the
FDR is controlled at level~$\alpha$ as in \eqref{equFDRcontrol}, we
aim to have a procedure that rejects a volume of hypotheses as large as
possible. In that sense, choosing a step-up procedure with $\beta
(x)=x$ always leads to a better step-up procedure than choosing $\beta
(x)$ of the form \eqref{condbeta}, because $\int_0^x u \,\mathrm
{d}\nu(u)\leq
x$. Hence, in Theorem~\ref{mainthm}, the PRDS assumption allows us to
get a result which is less conservative (i.e., rejecting more) than
under arbitrary dependencies.
Therefore, when we want to apply Theorem~\ref{mainthm}, an important
issue is to obtain,
if possible, the finite dimensional PRDS condition, see the examples of
Section~\ref{secappli}. When the PRDS assumption does not hold, we
refer to Blanchard and Roquain \cite{BR2008}
for an extended discussion on choices of the shape function $\beta$
of the form \eqref{condbeta} (which can be suitably adapted to the
uncountable case).

\subsection{Applications}\label{secappli}

\subsubsection{FDR control for testing the mean of a Gaussian process}
\label{gaussproc}

Consider the multiple testing setting of Example~\ref{exmeanprocess}.
More specifically, we consider here the particular case where
we observe $\{X_t,t\in[0,1]^d\}$ a Gaussian process with measurable
mean $\mu$, with unit variance and covariance function $\Sigma$.
Recall that the problem is to test for all $t\in[0,1]^d$ the
hypothesis $H_ t$: ``$\mu(t)\leq0$''. Taking for $\Lambda$ the
$d$-dimensional Lebesgue measure, the FDR control at level $\alpha$
of\vadjust{\goodbreak}
a step-up procedure of shape function $\beta$ and weight function $\pi
(h)=1$ can be rewritten as
%
%
\begin{equation}
\label{equFDRcontrolcontproc} \mbe\biggl[\frac{\Lambda(\{t \in[0,1]^d
\telque\mu (t)\leq0, \overline{\Phi}(X_t) \leq\alpha\beta(\hat
{r}(X))\} )} {\Lambda(\{t \in[0,1]^d \telque\overline{\Phi }(X_t)
\leq\alpha\beta(\hat{r}(X))\} )} \biggr] \leq\alpha,
\end{equation}
where $\overline{\Phi}$ is the upper-tail distribution function of a
standard Gaussian variable and 
\[
\hat{r}(X)=\max \bigl\{r \in[0,1] \telque\Lambda \bigl(\bigl\{t \in
[0,1]^d \telque\overline{\Phi}(X_t) \leq\alpha\beta(r)
\bigr\} \bigr)\geq r \bigr\} .
\]

Thus, Theorem~\ref{mainthm} and Lemma~\ref{lem-PRDS-Gauss} entail the
following result.

%
\begin{corollary}\label{corGaussudep}
For any jointly measurable Gaussian process $\{X_t\}_{t\in[0,1]^d}$
over $[0,1]^d$ with a measurable mean $\mu$ and unit variances,
the FDR control \eqref{equFDRcontrolcontproc} holds in either of the
two following cases:
\begin{itemize}
\item$\beta(x)=x$ and the covariance function of the process is
coordinates-wise non-negative, that is,
satisfies \eqref{equsigmapos};
\item$\beta$ is of the form \eqref{condbeta}, under no assumption on
the covariance function.
\end{itemize}
\end{corollary}

For instance, any Gaussian process with continuous paths is measurable
and thus can be used in Corollary~\ref{corGaussudep}.
More generally, Lemma~S-1.2 states that any Gaussian
process with a covariance function $\Sigma(t,t')$ such that
\[
\forall t \in[0,1]^d,\qquad \lim_{t'\rightarrow t}\Sigma \bigl(t',t
\bigr) = \Sigma(t,t) \quad\mbox{and}\quad \lim_{t'\rightarrow t}\Sigma \bigl(t',t'
\bigr) = \Sigma(t,t) =1
\]
has a measurable modification and hence can be used in Corollary~\ref
{corGaussudep}.

\subsubsection{FDR control for testing the signal in a Gaussian white
noise model}

We continue Example~\ref{exwhitenoise},
in which we observe the Gaussian process $X$ defined by $X_g = \int_0^1 g(t) f(t) \,\mathrm{d}t + \int_0^1 g(t) \,\mathrm{d}B_t$, $g\in
L^2([0,1])$, where $B$
is a Wiener process on $[0,1]$ and $f\in C([0,1])$ is a continuous
signal function.
Remember that we aim at testing $H_ t$: ``$f(t)\leq0$'' for any $t\in
[0,1]$, using the integration
of the process against a smoothing kernel $K_t$.
Assuming condition \eqref{eqregcondf} holds,
the $p$-value process is obtained via \eqref{equpvalueBBG} as
$p_t(X)=\overline{\Phi}^{-1}(Y_t)$, where $Y_t=v_{K,t}^{-1/2} (
X_{K_t} - \delta_{t,\eta} c_{K,t})$ is a Gaussian process.
Applying Lemma~\ref{lem-PRDS-Gauss}, we can prove that the $p$-value
process defined by \eqref{equpvalueBBG} is finite dimensional strong
PRDS (on any subset) by checking that the covariance function of
$(Y_t)_t$ has non-negative values: the latter holds because the kernel~$K$
has been taken non-negative and $\forall t,s$, $\Cov(Y_t,Y_{s})=c
\int_0^1K((t-u)/\eta)K((s-u)/\eta) \,\mathrm{d}u$, for a
non-negative constant
$c$. As a consequence, Theorem~\ref{mainthm} shows that a step-up
procedure using $\beta(x)=x$ controls the FDR.\looseness=-1

To illustrate this result, let us consider a simple particular case
where the kernel $K$ is rectangular, that is, $K(s)={\mbf{1}\{|s|\leq
1\}}/2$
and $f$ is $L$-Lipschitz. Also, to avoid the boundary effects due to
the kernel smoothing, we assume that the observation $X$ is made
against functions of $L^2([-1,2])$ while the test of $H_ t$:
``$f(t)\leq0$'' has only to be performed for $t\in[0,1]$ only.
In that case, for $t\in[0,1]$, $\delta_{t,\eta}=L\eta$,
$c_{K,t}=\eta$, $v_{K,t}=\eta/2$,
so that $Y_t=(2\eta)^{-1/2}(Z_{t+\eta}-Z_{t-\eta} -L\eta^2)$.
Therefore, the following statement holds.

%
\begin{corollary}\label{corBBG}
Let us consider the Gaussian process $Z_t = \int_{-1}^t f(s) \,\mathrm
{d}s +
B_t$, $t\in[-1,2]$, where~$B$ is a Wiener process on $[-1,2]$ and $f$
is a $L$-Lipschitz function on $[-1,2]$ ($L>0$). Let $\eta\in(0,1]$
and $Y_t=(2\eta)^{-1/2}(Z_{t+\eta}-Z_{t-\eta} -L\eta^2)$. Denote
the Lebesgue measure on $[0,1]$ by $\Lambda$. Consider the volume
rejection of the step-up procedure using $\pi(t)=1$ and $\beta(x)=x$,
that is,
\[
\hat{r}(X)=\max \bigl\{r \in[0,1] \telque\Lambda \bigl(\bigl\{t \in [0,1]
\telque \overline{\Phi}(Y_t) \leq\alpha r\bigr\} \bigr)\geq r \bigr\}
,
\]
where $\overline{\Phi}$ denotes the upper-tail distribution function
of a standard Gaussian variable. Then the following FDR control holds:
%
%
\begin{equation}
\label{equFDRcontrolcontproc-BBG} \mbe\biggl[\frac{\Lambda(\{t \in[0,1]
\telque f(t)\leq 0, \overline{\Phi}(Y_t) \leq\alpha\hat{r}(X)\}
)} {\Lambda (\{t \in[0,1] \telque\overline{\Phi}(Y_t) \leq\alpha
\hat{r}(X)\} )}\biggr] \leq
\alpha.
\end{equation}
\end{corollary}

\subsubsection{FDR control for testing the c.d.f.}

Consider the testing setting of Example~\ref{excdf} where we aim at
testing whether ``$F(t)\leq F_0(t)$'' for any $t$ in an interval
$I\subset\R$.
Lemma~\ref{PRDS-cdf} states that the $p$-value process defined by
\eqref{pvaluefit} is finite dimensional weak PRDS (on any subset).
As a consequence, Theorem~\ref{mainthm} applies and leads to a control
of the FDR.

For instance, let us consider the simple case where $I=[0,1]$,
$F_0(t)=t$ and $\Lambda$ is the Lebesgue measure on $[0,1]$.
In this case,
for any $k\in\{1,\ldots,m\}$, the function $G_t(k)=\P(Z_t\geq k)$,
with $Z_t\sim\mathcal{B}(m,t)$, is continuous increasing in the
variable $t\in[0,1]$.
Moreover, for any $t\in(0,1)$, the function $G_t(k)$ is decreasing in
$k=0,\ldots,m$.
Therefore, denoting $0=X_{(0)}\leq X_{(1)}\leq\cdots\leq X_{(m)}\leq
X_{(m+1)}=1$ the order statistics of $X_1,\ldots,X_m$, the $p$-value
process $t\mapsto p_t(X)=G_t(|\{1\leq i \leq m\telque X_i\leq t\}|)$ is
equal to $1$ on $[0,X_{(1)})$, is increasing on each interval
$(X_{(j)},X_{(j+1)}]$, $j=1,\ldots,m$, and is left-discontinuous and
right-continuous in each $X_{(j)}$, $1\leq j\leq m$, with a left limit
larger than $p_{X_{(j)}}(X)=G_{X_{(j)}}(j)$ (see Figure~\ref{figtestcdf}).

As a consequence, for any threshold $u\in(0,1)$, we obtain the
following relation for the Lebesgue measure $\gamma(u)$ of the level
set $\{t\in[0,1]\telque p_t(X)\leq u \}$:
%
%
\begin{eqnarray}\label{equ-cdf-nu}
\gamma(u)&=& \sum_{j=0}^m {\mbf{1}\bigl
\{G_{X_{(j)}}(j)\leq u\bigr\}} \Lambda \bigl( \bigl\{t \geq X_{(j)}
\telque G_t(j)\leq u \mbox{ and } t<X_{(j+1)} \bigr\} \bigr),
\nonumber
\\[-8pt]
\\[-8pt]
\nonumber
& =& \sum_{j=0}^m
\bigl(X_{(j+1)}\wedge t_j(u) - X_{(j)} \bigr)_+,
\end{eqnarray}
where $t_j(u), j=0,\ldots,m$ is the unique solution of the equation $G_t(j)=u$,
which can be easily computed numerically. Choosing for simplicity a uniform
weighting $\pi(x)\equiv1$, the choice of the rejection threshold
given by the linear step-up procedure is then $\wh{u}=\alpha\wh{r}$,
where $\wh{r}$ is the largest solution of the equation $\gamma(\alpha r)=r$.
To sum up, we have shown the following result.

%
\begin{corollary}\label{cortestcdf}
Let $X=(X_1,\ldots,X_m)$ be a vector of $m$ i.i.d. real random
variables of common continuous c.d.f. $F$.
Consider $(p_t(X))_{t\in[0,1]}$ the $p$-value process $p_t(X)=G_t(|\{
1\leq i \leq m\telque X_i\leq t\}|)$ for $G_t(k)=\P(Z_t\geq k)$,
where\vadjust{\goodbreak}
$Z_t$ is a binomial variable of parameters $(m,t)$. Assume that the
hypothesis space $[0,1]$ is endowed with the Lebesgue measure $\Lambda
$. Consider the volume rejection of the step-up procedure given by
%
%
\begin{equation}
\label{equrchapfit} \wh{r}(X)=\max\bigl\{r \in[0,1] \telque\gamma (\alpha r)
\geq r\bigr\},
\end{equation}
where $\gamma(\cdot)$ is defined by \eqref{equ-cdf-nu}. Then the
following FDR control holds:
%
%
\begin{equation}
\label{equFDRcontrolfit} \mbe\biggl[\frac{\Lambda(\{t\in[0,1] \telque
F(t)\leq t, p_t(X) \leq\alpha\hat{r}(X)\} )} {\Lambda(\{t \in
[0,1] \telque p_t(X) \leq\alpha\hat{r}(X)\} )} \biggr] \leq\alpha.
\end{equation}
\end{corollary}

\subsubsection{FDR control for testing the intensity of a Poisson
process}\label{secpois-ex}

Let us consider the testing setting of Example~\ref{exPoisson}.
Lemma~\ref{PRDS-pois} states that the $p$-values process is finite
dimensional strong PRDS (on any subset). Thus, it is also finite
dimensional weak PRDS (on any subset) by Lemma~S-2.2, and
Theorem~\ref{mainthm} leads to a control of the FDR.

%
\begin{figure}[t]
\vspace*{-3pt}
\includegraphics[scale=0.985]{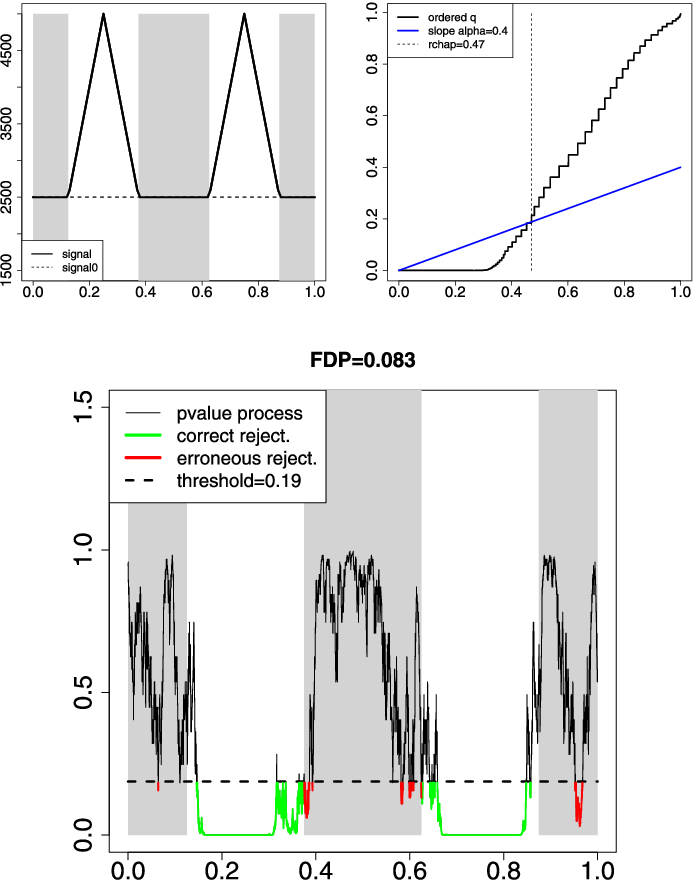}\vspace*{-3pt}

\caption{Several plots versus $t\in[0,1]$. Top left: $\lambda(t)$
(solid) and $\lambda_0$ (dashed).
Top right: $q_{\sigma(k)}(X)$ and $\alpha\sum_{\l=0}^{k} w_{\l} $
in function of $\sum_{\l=0}^{k} w_{\l}$,
for $k=1,\ldots,m$. Bottom: $p$-value process $p_t(X)$ defined by
\protect\eqref{pvalueprocesssimple-pois}. $\eta=0.015$, $\alpha=0.4$.
The grey areas indicate regions where the null hypotheses are true.}
\label{figtestpois}\vspace*{-4pt}
\end{figure}

Now, we aim at finding a closed formula for the linear
step-up procedure (\mbox{$\beta(x)=x$}) using the $p$-value process $(p_t(X))_t$.
Let us consider the particular case where the benchmark intensity
$\lambda_0(\cdot)$ is constantly equal to some $\lambda_0>0$ while
$\lambda(\cdot)$ is $L$-Lipschitz.
Also, to avoid the boundary effects, assume that the process
$(N_t)_{t}$ is observed for $t\in[-1,2]$ while $H_ t$: ``$\lambda
(t)\leq\lambda_0$'' is tested only for $t\in[0,1]$. In this case, the
$p$-value process is simply given by
%
%
\begin{equation}
\label{pvalueprocesssimple-pois} p_t(X)=G (N_{t+\eta}-N_{t-\eta}
),
\end{equation}
where for any $k\in\mathbb{N}$, $G(k)$ denotes $\mbp[Z\geq k]$ for
$Z$ a Poisson distribution of parameter
$2\eta\lambda_0+L \eta^2$
(note that $G(\cdot)$ is independent of $t$).
Consider the jumps $\{T_j\}_j$ of the process $(N_t)_{t\in[-1,2]}$ and
the set $\cS=\{s_i\}_{2\leq i\leq m}$ of the distinct and ordered
values of the set
$\bigcup_j\{T_j-\eta,T_j+\eta\}\cap(0,1)$. Moreover, we let $s_1=0$
and $s_{m+1}=1$.
Next, since the $p$-value process is constant on each interval
$[s_i,s_{i+1})$, $1\leq i\leq m$, we have for any $u\geq0$,
\begin{eqnarray*}
\Lambda \bigl( \bigl\{t\in[0,1]\telque p_t(X)\leq u \bigr\} \bigr)&=&
\sum_{i=1}^{m} (s_{i+1}-s_i)
{\mbf{1}\bigl\{p_{s_i}(X)\leq u\bigr\}}
\\
&=&\sum_{k=1}^{m} w_{k} {\mbf{1}
\bigl\{q_{\sigma(k)}(X)\leq u\bigr\}},
\end{eqnarray*}
where we let $q_{i}(X)=p_{s_i}(X)$, where $\sigma$ is a permutation of
$\{1,\ldots,m\}$ such that $q_{\sigma(1)}\leq\cdots\leq q_{\sigma
(m)}$ and where $w_{k}=s_{\sigma(k)+1}-s_{\sigma(k)}>0$ can be
interpreted as a ``weighting'' associated to $q_{\sigma(k)}$.
As a consequence, we get
%
%
\begin{eqnarray}
\label{rchap-pois} \wh{r}(X)&=&\max\Biggl\{r \in[0,1] \telque\sum
_{\l
=0}^{m} w_{\l} {\mbf{1}\bigl
\{q_{\sigma(\l)}(X) \leq\alpha r\bigr\}} \geq r\Biggr\}
\nonumber
\\[-8pt]
\\[-8pt]
\nonumber
&=&\max\Biggl\{\sum_{\l=0}^{k}
w_{\l}, \mbox{ for }k\in\{0,\ldots,m\} \mbox{ s.t. }
q_{\sigma(k)}(X)\leq\alpha\sum_{\l=0}^{k}
w_{\l}\Biggr\},
\end{eqnarray}
because since $\wh{r}(X)$ is a maximum, it is of the form $\sum_{\l
=0}^{k} w_{\l}$, $k\in\{0,\ldots,m\}$.
Note that we should let $q_{\sigma(0)}=0$ and $w_0=0$ to cover the
case $\wh{r}(X)=0$.
Relation \eqref{rchap-pois} only involves a finite number of
variables. Thus, $\wh{r}(X)$ can be easily computed in practice.
This is illustrated in Figure~\ref{figtestpois}.

We have proved the following result.

%
\begin{corollary}\label{cortest-pois}
Let $X=(N_t)_{t\in[-1,2]}$ be a Poisson process with an intensity
$\lambda\dvtx [-1,2]\rightarrow\R^+$ $L$-Lipschitz ($L>0$) and let
$\lambda_0>0$. For $\eta\in(0,1]$, consider the $p$-value process $\{
p_t(X)\}_{t\in[0,1]}$ given by \eqref{pvalueprocesssimple-pois}.
Assume that the hypothesis space $[0,1]$ is endowed with the Lebesgue
measure $\Lambda$.
Then $\wh{r}(X)$ defined by \eqref{rchap-pois} satisfies the following:
%
%
\begin{equation}
\label{equFDRcontrolpois} \mbe\biggl[\frac{\Lambda(\{t\in[0,1] \telque
\lambda (t)\leq\lambda_0, p_t(X) \leq\alpha\hat{r}(X)\} )} {
\Lambda(\{t \in[0,1] \telque p_t(X) \leq\alpha\hat {r}(X)\} )} \biggr] \leq\alpha.
\end{equation}
\end{corollary}

To illustrate Corollary~\ref{cortest-pois}, Figure~\ref{figtestpois}
displays the case where $\lambda(t)$ is a truncated triangular signal.
The choice of the bandwidth $\eta$ has been made manually, see
Section~\ref{secdiscusseta} for a discussion on this point.

%
\begin{remark} \label{rempois}
Up to increase the set $\cS=\{s_i\}_{ i}$ so that $t\mapsto{\mbf{1}\{
\lambda(t)\leq\lambda_0\}}$ is constant over each $[s_i,s_{i+1})$,
the FDR control \eqref{equFDRcontrolpois} can be rewritten as
%
%
\begin{equation}
\label{equFDRcontrolpoisdiscrete} \mbe\biggl[\frac{ \sum_{i=1}^m
(s_{i+1}-s_i){\mbf{1}\{\lambda(s_i)\leq \lambda_0\}} {\mbf{1}\{
p_{s_i}(X)\leq\al \hat{r}(X)\}} } { \sum_{i=1}^m (s_{i+1}-s_i)
{\mbf{1}\{p_{s_i}(X)\leq\al \hat{r}(X)\}} } \biggr] \leq
\alpha.
\end{equation}
Hence, the procedure \eqref{rchap-pois} appears as controlling the
discrete FDR-weighting on $\{1,\ldots,m\}$ where the weight for
rejecting ``$\lambda(s_i)\leq\lambda_0$'' is $(s_{i+1}-s_i)$ and
where the initial $p$-values are $q_i(X)=p_{s_i}(X)$.
The rationale behind this is that if $q_i(X)=p_{s_i}(X)$ is below $\hat
{r}(X)$, then so are all $p_t(X)$, $t\in[s_i,s_{i+1})$. Hence, a
rejection for a $p$-value $q_i(X)=p_{s_i}(X)$ accounts for the length
of the entire interval in the FDR. From an intuitive point of view,
this means that the type I error importance in the FDR is larger for
``isolated'' points of the process.
This bears some similarity with discrete multiple testing with
weights,
Benjamini and Hochberg \cite{BH1997} and Blanchard and Roquain~\cite
{BR2008}, but those results would not apply
here since the weights
themselves are data-dependent.
\end{remark}

\section{\texorpdfstring{Proof of Theorem~\protect\ref{mainthm}}{Proof of Theorem 4.1}}\label{secproofmainth}

\subsection{Two conditions for controlling the FDR }\label{secFDRapproach}

Similarly to Proposition~2.7 of Blanchard and Roquain \cite{BR2008}
(which we refer
to as BR08 for short from now on), we can prove that the
FDR control
$\FDR(R,P) \leq\al\Pi(\cH_0(P)) $
holds true for any $P\in\cP$ as soon as the two following
sufficient conditions hold for any $P \in\cP$:
\begin{itemize}
\item the multiple testing procedure $R$ satisfies the
``self-consistency condition''
%
%
{\renewcommand{\theequation}{$\mathrm{SC}(\alpha,\pi,\beta)$}
\begin{equation}
\label{SCC} R(x) \subset \bigl\{ h \in\cH\telque p_h(x) \leq\alpha
\pi(h) \beta \bigl( \Lambda \bigl(R(x) \bigr) \bigr) \bigr\} \qquad\mbox{for } P
\mbox{-almost all } x \in\cX%
\end{equation}}
\item for any $h \in\cH_0(P)$ the couple of
real random variables $(U_h,V)\bydef(p_h(X),\Lambda(R(X)))$ satisfies
the ``dependence control condition''
%
%
{\renewcommand{\theequation}{$\mathrm{DC}(\beta)$}
\begin{equation}
\label{DCC} \forall c>0,\qquad \mbe\biggl[\frac{{\mbf{1}\{U_h \leq c \beta(V)\}
}}{V} {\mbf{1}\{V>0\}}\biggr]
\leq c . 
\end{equation}}
\end{itemize}

The proof is as follows: by definition and by using Fubini's theorem,
we have
\begin{eqnarray*}
\FDR(R,P)&= &\mbe\biggl[\frac{\Lambda(R\cap\cH_0)}{\Lambda(R)} {\mbf {1}\bigl\{\Lambda(R)>0\bigr\}}
\biggr]
\\
&=& \mbe\biggl[\int_{h\in\cH_0} \frac{{\mbf{1}\{h\in R\}}}{\Lambda(R)} {\mbf{1}\bigl\{
\Lambda(R)>0\bigr\}}\,\mathrm{d}\Lambda(h)\biggr]
\\
&=& \int_{h\in\cH_0} \mbe\biggl[\frac{{\mbf{1}\{h\in R\}}}{\Lambda
(R)}{\mbf{1}\bigl\{
\Lambda(R)>0\bigr\}}\biggr] \,\mathrm{d}\Lambda(h)
\\
&\leq&\int_{h\in\cH_0} \mbe\biggl[\frac{{\mbf{1}\{p_h\leq\alpha\pi
(h)\beta(\Lambda(R))\}}}{\Lambda(R)}{\mbf{1}\bigl\{
\Lambda(R)>0\bigr\}}\biggr] \, \mathrm{d}\Lambda(h)
\\
&\leq&\alpha\int_{h\in\cH_0} \pi(h) \,\mathrm{d}\Lambda(h),
\end{eqnarray*}
where we have used the shortened notation $R$ for $R(X)$ and $p_h$ for
$p_h(X)$, and used successively conditions \eqref{SCC} and \eqref{DCC}
for the two above inequalities. Observe that the use of Fubini's
theorem is granted by the measurability assumption of
Definition~\ref{defMTP}.

Therefore, to obtain the FDR bound of Theorem~\ref{mainthm} in each
case, we simply have to check conditions
\eqref{SCC} and \eqref{DCC} in the different settings.

\subsection{\texorpdfstring{Any step-up procedure satisfies \protect\eqref{SCC}}
{Any step-up procedure satisfies (SC(alpha, pi, beta))}}
\label{secstep-upproof}

From the definition of a step-up procedure, for all $\varepsilon>0$,
we have $\Lambda(L_\Delta(\wh{r}))\leq\Lambda(L_\Delta(\wh
{r}+\varepsilon)) < \wh{r}+\varepsilon$. This entails that $\wh{r}$
satisfies $\Lambda(L_\Delta(\wh{r})) =\wh{r}$. Hence, the step-up
procedure $R$ satisfies \ref{SCC} with equality.

We now check that any step-up procedure is a multiple testing
procedure, that is, that $(\omega,h) \mapsto{\mbf{1}\{h \in
R(X(\omega ))\}}={\mbf{1}\{p_h(X(\omega))\leq\al\pi(h) \beta(\wh
{r}(X(\omega )))\}}$ is (jointly) measurable. From \textup{\eqref
{hypmesomega}} and since
$\beta$ and $\pi$ are measurable, it is enough to check that $\omega
\mapsto\wh{r}(X(\omega))$ is measurable. For any $x\in X(\Omega)$,
let us consider the function
%
%
\[
f\dvtx r\in\mathbb{R}^+\mapsto\Lambda \bigl(L_\Delta(x,r)
\bigr)= \int_{\cH} {\mbf{1}\bigl\{p_h(x)\leq\al
\pi(h)\beta(r)\bigr\}}\,\mathrm{d} \Lambda(h).
\]
We observe that $f$ is right-continuous and non-decreasing (because
$\beta$ is) and bounded, and that $\wh{r}=\max\{r \geq0 \telque
f(r) \geq r\}$.
Applying Lemma~S-2.3, we deduce that for\break  any~$x\in X(\Omega)$,
%
%
\setcounter{equation}{22}
\begin{eqnarray}
\label{decompSU} \wh{r}(x) &=& \inf_{\varepsilon>0, \varepsilon\in
\mathbb{Q}} \sup \bigl\{r\in\mathbb{Q}^+
\telque\Lambda \bigl(L_\Delta(x,r) \bigr)\geq r- \varepsilon \bigr\}
\nonumber
\\[-8pt]
\\[-8pt]
\nonumber
&=& \inf_{\varepsilon>0, \varepsilon\in\mathbb{Q}} \sup_{r\in
\mathbb{Q}^+} \bigl(r {\mbf{1}\bigl\{\Lambda
\bigl(L_\Delta(x,r) \bigr)\geq r- \varepsilon\bigr\}} \bigr).
\end{eqnarray}
Since from \textup{\eqref{hypmesomega}}, for all $\varepsilon>0$,
$\varepsilon\in\mathbb{Q}$ and $r \in\mathbb{Q}^+ $, the function
\begin{eqnarray*}
\omega\mapsto& r {\mbf{1}\bigl\{\Lambda \bigl(L_\Delta \bigl(X(\omega
),r \bigr) \bigr)\geq r- \varepsilon\bigr\}} = r {\mbf{1}\bigl\{\Lambda \bigl(
\bigl\{h\in\cH\telque p_h \bigl(X(\omega) \bigr)\leq\alpha\pi (h)
\beta(r) \bigr\} \bigr)\geq r- \varepsilon\bigr\}}
\end{eqnarray*}
is measurable, expression \eqref{decompSU} implies that $\omega
\mapsto\wh{r}(X(\omega))$ is measurable. Hence, a step-up procedure
satisfies the measurability requirements of Definition~\ref{defMTP}.

\subsection{\texorpdfstring{Conditions implying \protect\eqref{DCC}}
{Conditions implying (DC(beta))}}\label{secdcc}

We use the following lemma which was proved in (BR08) (see Lemma~3.2, items
(ii, iii) therein):

%
\begin{lemma}
\label{UVlemma}
Let $(U,V)$ be a couple of non-negative random variables such that $U$
is stochastically
lower bounded by a uniform variable on $[0,1]$, that is, $\forall t\in
[0,1], \P(U\leq t)\leq t$. Then the dependence control condition
\ref{DCC} is satisfied by $(U,V)$ in either one of the following situations:

\begin{longlist}[(ii)]
\item[(i)] $\beta(x)=x$ and
%
%
\begin{equation}
\label{UV-nondec} \forall r\in\mathbb{R}^+,\qquad u\mapsto\Proba(V < r |
U \leq u)\qquad
\mbox{is non-decreasing on } \bigl\{u\telque P(U\leq u)>0 \bigr\} .
\end{equation}

\item[(ii)] The shape function $\beta$ is of the form
\eqref{condbeta}.
\end{longlist}
\end{lemma}

The point (ii) above, together with the results of the two previous sections,
establishes point 2 of Theorem~\ref{mainthm}. To establish point 1 and
finish the proof, we have to prove that \eqref{UV-nondec} holds in the
finite dimensional weak PRDS dependence context, which is done in the
following proposition:

%
\begin{proposition}\label{corWeakPRDSimpliesfdrcontrol}
Assume that the $p$-values process $\mathbf{p}=(p_h,h\in\cH)$ is
finite dimensional weak PRDS on $\cH_0(P)$ for any $P\in\cP$.
Consider $R$ the step-up procedure defined by Definition~\ref
{defcontstepup} with $\beta(x)=x$. Then for any $P\in\cP$,
for any $h\in\cH_0(P)$, the couple of variables $(U_h,V)=(p_h,\Lambda
(R))$ satisfies \eqref{UV-nondec} and thus
\ref{DCC} holds for $\beta(x)=x$.
\end{proposition}

\begin{pf}
In the above statement and the present proof, we use the shortened
notation $R,p_h$, and $L_\Delta(r)$ for
the random quantities $R(X),p_h(X)$, and $L_\Delta(X,r)$, respectively.
The goal of the proof is to establish \eqref{UV-nondec}, that is for
any $h_0 \in\cH_0$ ($h_0$ is assumed to
be fixed in $\cH_0$ in the rest of the proof), for any
$t$, and $0\leq u\leq u'$ with $\P(p_{h_0}\leq u)>0$:
\[
\mbp\bigl[\Lambda(R) < t | p_{h_0} \leq u\bigr]
\leq\mbp\bigl[\Lambda(R) < t | p_{h_0} \leq u'
\bigr] ;
\]
From Definition~\ref{defcontstepup}, the real random variable
$\Lambda(R)$ can be rewritten as $\Lambda(R)=\wh{r}=\max\{r\telque
f(r)\geq r\}$ with $f\dvtx r \mapsto\Lambda(L_\Delta(r))$.
Furthermore, denoting $G_u =
\frac{{\mbf{1}\{p_{h_0} \leq u\}}}{\mbp[p_{h_0} \leq u]} $,
we are equivalently aiming at proving that for any $t$ and $0 \leq u
\leq u'$ with $\P(p_{h_0}\leq u)>0$:
%
%
\begin{equation}
\label{dagoal} \mbe\bigl[{\mbf{1}\{\wh{r} < t\}} G_u\bigr] \leq\mbe
\bigl[{\mbf {1}\{\wh{r} < t \}} G_{u'}\bigr] .
\end{equation}
By using Lemma~\ref{approxcontstepup} (and the notation therein),
there exists a fixed sequence of finitely supported measures $\Lambda_n$ on $\cH$ such that, denoting
$
\wh{r}_{n,k}=\max\{r\geq0\telque\Lambda_n(L_\Delta(r))\geq
r-k^{-1} \},
$
it holds that
%
%
\begin{equation}
\label{equapprox} \wh{r}=\lim_{k\rightarrow\infty} \wh{r}^+_k =
\lim_{k\rightarrow
\infty} \wh{r}^-_k\qquad \mbox{almost surely},
\end{equation}
where we let $\wh{r}^+_k=\limsup_{n\rightarrow\infty} \wh
{r}_{n,k}$ and $\wh{r}^-_k=\liminf_{n\rightarrow\infty} \wh{r}_{n,k}$.

Let $\cS_n$ be the (finite) support of $\Lambda_n$ and $\cS'_n = \cS_n \cup\{h_0\}$.
Writing $\wh{r}_{n,k}$ as a function of the
finite $p$-value set $\{p_{h}, h\in\cS'_n\}$, the function $\wh
{r}_{n,k} \dvt\mbf{z}=(z_h)_{h\in\cS'_n} \in[0,1]^{\cS'_n}
\mapsto
\wh{r}_{n,k}(\mbf{z}) $ is measurable (where the space $[0,1]^{\cS
'_n}$ is endowed with the standard product Borel $\sigma$-field), and
is additionally non-increasing in each $p$-value.
Hence, the set $\{\mbf{z}=(z_h)_{h\in\cS'_n} \telque\wh
{r}_{n,k}(\mbf{z}) <t+k^{-1}\}$ is a non-decreasing measurable subset
of $[0,1]^{\cS'_n}$.
Using that the $p$-value process $\mbf{p}=(p_h, h\in\cH)$ is finite
dimensional weak PRDS on $\cH_0$, the $p$-values $(p_h, h\in\cS'_n)$
are PRDS on $\cH_0\cap\cS'_n$, which implies that
for any $t\geq0$ and $u \leq u'$ with $\P(p_{h_0}\leq u)>0$,
%
%
\begin{equation}
\mbe\bigl[{\mbf{1}\bigl\{\wh{r}_{n,k}-k^{-1} < t\bigr\}}
G_u\bigr] \leq\mbe\bigl[{\mbf{1}\bigl\{\wh{r}_{n,k}-k^{-1}
< t \bigr\}} G_{u'}\bigr] .\label {equfiniteequ}
\end{equation}
Now, to prove \eqref{dagoal}, it suffices to carefully make $n$ and
$k$ tend to infinity.
By Fatou's lemma and by \eqref{equfiniteequ}, we have for all $k\geq
1$:
\begin{eqnarray*}
\mbe\Bigl[\liminf_n {\mbf{1}\bigl\{\wh{r}_{n,k} -
k^{-1}< t\bigr\}} G_u \Bigr] &\leq& \liminf_n
\mbe\bigl[{\mbf{1}\bigl\{\wh{r}_{n,k} - k^{-1}< t\bigr\}}
G_u\bigr]
\\
& \leq& \limsup_n \mbe\bigl[{\mbf{1}\bigl\{\wh{r}_{n,k} -
k^{-1}< t\bigr\}} G_{u'} \bigr]
\\
& \leq& \mbe\Bigl[\limsup_n{\mbf{1}\bigl\{\wh{r}_{n,k} -
k^{-1}< t\bigr\}} G_{u'} \Bigr]
.
\end{eqnarray*}
Notice that the following inclusions of events hold:
$\{\wh{r}^+_{k} < t + k^{-1} \}\subset\liminf_n \{\wh{r}_{n,k} <
t+k^{-1}\}$, $\limsup_n \{\wh{r}_{n,k} < t+k^{-1}\}\subset\{\wh
{r}^-_{k} \leq t + k^{-1} \}$. Hence, we obtain for all $k$:
\[
\mbe\bigl[{\mbf{1}\bigl\{\wh{r}^+_{k} - k^{-1}< t\bigr\}}
G_u\bigr] \leq\mbe\bigl[{\mbf{1}\bigl\{\wh{r}^-_{k} -
k^{-1}\leq t\bigr\}} G_{u'}\bigr] .
\]
Then, if $t$ is such that $\mbp[\wh{r}= t]=0$, the above expression
can be rewritten as
\begin{eqnarray*}
\mbe\bigl[{\mbf{1}\bigl\{\wh{r}^+_{k} - k^{-1}< t\bigr\}}
G_u {\mbf{1}\{\wh {r}\neq t\}}\bigr] \leq\mbe\bigl[{\mbf{1}\bigl\{
\wh{r}^-_{k} - k^{-1} \leq t\bigr\}} G_{u'}{
\mbf{1}\{\wh{r}\neq t\}}\bigr] .
\end{eqnarray*}
We now let $k\rightarrow\infty$ in the above expression by using
\eqref{equapprox} and the dominated convergence theorem: for any $u
\leq u'$ with $\P(p_{h_0}\leq u)>0$, and any $t \notin D\bydef\{s
\geq0 \telque\mbp[\wh{r}= s]>0\}$, we have
%
%
\begin{equation}
\mbe\bigl[{\mbf{1}\{\wh{r}< t\}} G_u\bigr] \leq\mbe\bigl[{\mbf{1}\{
\wh{r}< t\}} G_{u'}\bigr]. \label{proof-element}
\end{equation}
Since the above expectations may be interpreted as (conditional)
probabilities, the LHS and RHS in \eqref{proof-element} are
left-continuous functions of $t$. Using that $\R^+\cap D^c$ is dense
in $\R^+$ (because $D$ is at most countable), we obtain that \eqref
{proof-element} holds for any $t$.
Finally, the condition \eqref{DCC} comes from Lemma~\ref{UVlemma}.
\end{pf}

\subsection{Finite approximation of step-up procedures}\label
{secfiniteapprox}\label{seccomput-su}

As usual, to lighten notation
$R,p_h, L_\Delta(r),\wh{r}$ denote the random quantities
$R(X),p_h(X),\break L_\Delta(X,r),\wh{r}(X)$. The following result shows how
to derive the continuous step-up procedure (see Definition~\ref
{defcontstepup}) from a limit of finite step-up procedures. It is used
in the proof of Proposition~\ref{corWeakPRDSimpliesfdrcontrol}.

%
\begin{lemma}\label{approxcontstepup}
Consider the step-up procedure $R=L_\Delta(\wh{r})$ on $\cH$ using
$\Lambda$ and with $\wh{r}$ defined in Definition~\ref{defcontstepup}.
Then there exists a sequence of finitely supported measures $\Lambda_n$ on $\cH$ such that, denoting
\[
\wh{r}_{n,k}=\max \bigl\{r\geq0\telque\Lambda_n
\bigl(L_\Delta(r) \bigr)\geq r-k^{-1} \bigr\},
\]
we have
\[
\wh{r}=\lim_{k\rightarrow\infty} \Bigl(\limsup_{n\rightarrow
\infty} \wh{r}_{n,k}
\Bigr) =\lim_{k\rightarrow\infty} \Bigl( \liminf_{n\rightarrow
\infty} \wh{r}_{n,k}
\Bigr) \qquad\mbox{almost surely.}
\]
\end{lemma}

\begin{pf}
We start with the following observation. Consider $(\Lambda_n)$ some
sequence of measures on $\cH$ such that $\Lambda_n(\cH)\equiv M$.
For a fixed realization $x\in X(\Omega)$ of $X$, we consider $f\dvtx
r\in
\mathbb{R}^+\mapsto\Lambda(L_\Delta(x,r))$ and $f_{\Lambda_n}\dvtx
r\in\mathbb{R}^+\mapsto\Lambda_n(L_\Delta(x,r))$. Clearly, $f$
and $f_{\Lambda_n}$ are non-decreasing right-continuous functions.
Using Lemma~S-2.4, we conclude that the desired result
holds provided that, for $P$-almost all $x \in\cX$, $f_{\Lambda_n}$
converges uniformly to $f$ over $[0,M+1]$.

It remains thus to prove that there exists a sequence of finitely
supported measures
$\Lambda_n$ on $\cH$ such that for $P$-almost all $x \in\cX$,
%
%
\begin{equation}
\label{equtoprove} \limsup_{n \rightarrow\infty} \Bigl\{ \sup_{r\in[ 0,M+1]} \bigl|
\Lambda_n \bigl(L_\Delta(x,r) \bigr) - \Lambda
\bigl(L_\Delta(x,r) \bigr)\bigr | \Bigr\} = 0 .
\end{equation}
Denote $\cY$ the product space $\cH^{\mbn}$, endowed with the
product sigma-algebra. For $y\bydef(h_i)_{i \geq1} \in\cY$ some
sequence of hypotheses, denote
$\Lambda_n^{[y]}= M n^{-1} \sum_{i=1}^n \delta_{h_i}$ the suitably
scaled uniform atomic measure on $(h_1,\ldots,h_n)$.

Consider now $Y\bydef(H_i)_{i\geq1} \in\cY$ an i.i.d. sequence of
hypotheses drawn independently of $X$
according to the probability distribution $\Lambda/M$ on $\cH$.
Observe that for any fixed $x\in X(\Omega)$, $L_\Delta(x,r) = \{h\in
\cH\telque p_h(x)\leq\alpha\pi(h)\beta(r)\} = \{h\in\cH\telque
q(h,x)\leq\alpha\beta(r)\}$,
where
\[
q(h,x)\bydef\cases{ p_h(x)/\pi(h), & \quad$\mbox{if } \pi(h)>0;$
\vspace*{2pt}
\cr
0, & \quad$\mbox{if } \pi(h)=0 \mbox{ and } p_h(x)=0;$
\vspace*{2pt}
\cr
\alpha\beta(M+1)+1, &\quad  $\mbox{if } \pi(h)=0 \mbox{ and }
p_h(x)>0$. }
\]
Thus, applying the Glivenko--Cantelli theorem to the i.i.d. variables
$(q(H_i,x))_i$, we deduce that for any $x\in\cX(\Omega)$, $\zeta
(x,y) = \limsup_{n \rightarrow\infty} \sup_{r\in[0,M+1]} | \Lambda_n^{[y]}(L_\Delta(x,r)) - \Lambda(L_\Delta(x,r)) | = 0$ for
$P_Y$-almost all realizations $y$ of $Y$. Observe furthermore that
for any fixed $r$, the function
\[
(\omega,y) \in\Omega\times\cH^\mbn\mapsto\Lambda^{[y]}_n
\bigl(L_\Delta \bigl(X(\omega),r \bigr) \bigr) = M n^{-1} \sum
_{i=1}^n {\mbf{1}\bigl\{p_{h_i}
\bigl(X( \omega) \bigr) \leq\alpha\pi (h_i) \beta(r)\bigr\}}
\]
is a (jointly) measurable function of $(\omega,y)$
by assumption \textup{\eqref{hypmesomega}}.
The inside supremum in \eqref{equtoprove} can be restricted to
rational numbers since the functions involved are right-continuous.
Therefore, $(\omega,y)\mapsto\zeta(X(\omega),y)$ is a jointly
measurable function in its variables. By Fubini's theorem, this implies that
$\mbe_{X,Y}[\zeta(X,Y)]=0$; and thus also, for $P_Y$-almost all $y\in
\cY$, $\zeta(x,y) =0$ for $P$-almost all $x\in\cX$. Since an event
of probability $1$ is non-empty, there exists a fixed $y\in\cY$ such
that $\zeta(x,y) =0$ for $P$-almost all $x\in\cX$, which gives rise
to a sequence of finitely supported measures $\Lambda_n$ satisfying
\eqref{equtoprove}.\vspace*{-2pt}
\end{pf}

\section{Discussion}\label{secdiscuss}\vspace*{-2pt}

\subsection{FDR control for self-consistent, non-step-up
procedures}\vspace*{-2pt}

In some cases, for instance, after a discretization in
$r$ or under a global constraint over the admissible geometry of sets
of rejected hypotheses, the procedure of
interest may not be of the step-up form, while still satisfying the
more general condition~\eqref{SCC}
(called self-consistency, see
Section~\ref{secFDRapproach}). In that situation,
Theorem~\ref{mainthm} does not apply, because the procedure is not
step-up. We proved an extension of Theorem~\ref{mainthm} holding more
generally for (non-increasing)
self-consistent procedures, but point 1 of the theorem is established
only under a stronger
PRDS condition called general PRDS. (On the other hand, the fact that
point 2 of Theorem~\ref{mainthm}
remains valid under the more general condition~\eqref{SCC} is quite immediate.)
The general PRDS condition is defined in terms of the entire process
$X$ and not only
its finite dimensional projections. Therefore, it
is substantially more technical than
finite dimensional PRDS. In particular, it is an open question to
characterize when
does finite dimensional PRDS imply general PRDS (we provide some
sufficient conditions). For simplicity, we deferred
the corresponding study in part~II of the supplementary material
(Blanchard, Delattre and
Roquain \cite
{BDR2011-supp}).\vspace*{-2pt}

\subsection{Power and adaptive procedures}
\label{secdiscusseta}\vspace*{-2pt}

This work has focused on procedures ensuring control of the type I
error as measured by the FDR.
Under this constraint, one would like to maximize power.
We do not address this issue in the present work; a specific multiple
testing power criterion would
have to be defined to begin with, for instance the average number of
correct rejections.
We briefly discuss possible future directions in this regard,
in particular adaptivity properties with respect
to different types of underlying regularity structure.

\textit{Adaptivity of single tests}.
The power of a multiple testing procedure depends primarily on the
power of the underlying
single tests and $p$-values it is built upon. It\vadjust{\goodbreak} is of course desirable
to design individual tests
that are as powerful as possible in
the first place. While this issue actually pertains to the domain of
single hypothesis testing, and
is to this extent quite independent of the methodology studied here, we
briefly discuss this issue in
the light of the specific example of the Gaussian white noise model
$dZ_t = f(t)\,\mathrm{d}t + \sigma\,\mathrm{d}B_t$.
For designing a test of the hypothesis $f(t_0)=0$, we have assumed
known regularity of $f$ and
considered a test based on a simple kernel estimator. Could this be improved?

There is an abundance of literature on adaptive testing of a global
qualitative hypothesis on $f$ (the simplest example being
testing that $f$ is identically zero), where adaptation is understood
with respect to the (H\"{o}lder or Besov) regularity of the alternative
and separation from the null is generally measured in some $L^p$ norm.
This might give some hope that some form of regularity
adaptation is possible also
for testing the local hypothesis $f(t_0)=0$ (and the separation
distance $| f(t_0) |$),
but the situation is in fact quite different and
possibilities for this are severely limited.
This is in essence the same phenomenon as for the existence of
regularity-adaptive confidence intervals
for pointwise estimation of a function, as studied by Cai and Low \cite
{CaiLow04}.
We sketch the main arguments here. First, following the discussion in
D{\"u}mbgen and Spokoiny \cite{DS2001}, Section 2, observe that for
testing of $f_0$ against $g_0$,
the power of the optimal NP test is $\Phi(\sigma\|f_0-g_0\|_2)$,
where $\Phi$ is a non-decreasing function.
Thus, the optimal power of a composite null $H_0$ against an
alternative $H_1$ is upper bounded
(with equality if $H_0$ and $H_1$ are convex) by
$\Phi(\sigma\inf_{f\in\cH_0, g\in\cH_1}\|f-g\|_2)$. In the
case where $H_0=\{f\in\cF, f(t_0)=0\}$ and
$H_1=\{f\in\cF_1, f(t_0)=\eps\}$, where $\cF_1\subset\cF$ are
H\"{o}lder regularity classes,
Cai and Low~\cite{CaiLow04} (Example 1 there) establish that the rate
behavior as
$\eps\rightarrow0$ of this infimum distance
is determined by $\cF$ and not by $\cF_1$.
Therefore, no adaptation to the regularity of the alternative is
possible in this configuration, and it is
necessary to assume some {a priori} known regularity class $\cF$.
On the other hand, these authors show that adaptive confidence
intervals (and hence tests) exist in this setting
provided some additional shape restrictions, such as monotonicity, are
assumed to hold.

\textit{Adaptivity to $\Pi(\cH_0)$ and to the dependence structure}.
For multiple testing over a finite hypothesis space, recent research
has focused
on improving step-up procedures to take into account, on the one hand,
the (unknown) volume $\Pi(\cH_0(P))$ of true null
hypotheses -- which comes as a nuisance parameter reducing the
effective level, see \eqref{equFDRcontrol},
and on the other hand, the dependence structure of the $p$-values. Both
directions suggest further possible developments in the continuous
setting as well.

\begin{appendix}

\section*{Appendix: PRDS statements}\label{secPRDS-appendix}

%
\begin{lemma}\label{PRDS-cdf}
The $p$-value process $\mathbf{p}(X)=\{p_t(X),t\in I\}$ defined by
\eqref{pvaluefit} is finite dimensional weak PRDS (on any subset).
\end{lemma}

\begin{pf}
Let us consider a finite subset $(t_j)_{0\leq j \leq N-1}$ of $I$ and
$D$ a non-decreasing measurable subset of $[0,1]^{N}$.
Let us prove that the function
$
u\mapsto\mbp[\mathbf{p}(X) \in D | p_{t_0}(X)\leq u]
$
is non-decreasing on $\{u\in[0,1]\telque\P(p_{t_0}(X)\leq u)>0\}$.
If $F(t_0)\in\{0,1\}$, the result is trivial. We thus assume that
$F(t_0)\in(0,1)$, so that $\mathcal{U}_{t_0}=\{
G_{t_0}(k),k=m,m-1,\ldots,0\}$ contains only increasing points of
$(0,1]$. Without loss of generality, we only have to prove the
non-decreasing property for $u\in\mathcal{U}_{t_0}$.
Since $G_{t_0}$ is decreasing from $\{0,\ldots,m\}$ to $\mathcal
{U}_{t_0}$, we have $p_{t_0}(X)\leq G_{t_0}(k) \Longleftrightarrow
m\mathbb{F}_m(X,t_0) \geq k \Longleftrightarrow X_{(k)} \leq t_0$
(letting $X_{(0)}=-\infty$). We thus have to prove that for any $k$, $
1\leq k\leq m$,
\renewcommand{\theequation}{\arabic{equation}}
\setcounter{equation}{29}
\begin{equation}
\label{equfinalfitmulti}\mbp\bigl[(X_{(1)}, \ldots,
X_{(m)}) \in D' | X_{(k-1)}\leq t_0
\bigr] \geq \mbp\bigl[(X_{(1)}, \ldots, X_{(m)}) \in
D' | X_{(k)}\leq t_0\bigr],
\end{equation}
where $D'=\{x\in\mathbb{R}^m \telque(p_{t_j}(x))_{0\leq j\leq N-1}
\in D\}$ is a non-decreasing subset of $\mathbb{R}^m$ (because
$\mathbf
{p}$ is coordinate wise non-decreasing, that is, $x\leq x'\Rightarrow
\forall t, {p}_{t}(x)\leq{p}_{t}(x')$).
Using that the family of order statistics $\{X_{(i)}\}_i$ has positive
regression dependency (see Lemma~S-2.1), we derive
that the function $f(a,b)=\mbe[(X_{(1)}, \ldots, X_{(m)}) \in D' |
X_{(k-1)}=a,X_{(k)}=b]$ is non-decreasing in $a$ and $b$.
Therefore, denoting $\gamma=\mbp[X_{(k)}\leq t_0 | X_{(k-1)} \leq t_0]$, we get
\begin{eqnarray*}
\mbp\bigl[(X_{(1)}, \ldots, X_{(m)}) \in
D' | X_{(k-1)}\leq t_0\bigr] &= & \gamma\mbe
\bigl[f(X_{(k-1)},X_{(k)}) | X_{(k-1)}\leq
t_0,X_{(k)}\leq t_0\bigr]
\\
&&{}+ (1-\gamma) \mbe\bigl[f(X_{(k-1)},X_{(k)}) |
X_{(k-1)}\leq t_0<X_{(k)}\bigr]
\\
&\geq& \mbe\bigl[f(X_{(k-1)},X_{(k)}) | X_{(k-1)}\leq
t_0,X_{(k)}\leq t_0\bigr],
\end{eqnarray*}
which provides \eqref{equfinalfitmulti} and concludes the proof.
\end{pf}

%
\begin{lemma}\label{PRDS-pois}
The $p$-value process $\mathbf{p}(X)=\{p_t(X),t\in[0,1]\}$ defined
by \eqref{pvaluespois} is finite dimensional strong PRDS (on any subset).
\end{lemma}

\begin{pf}
Let $M_t=N_{(t+\eta)\wedge1}-N_{(t-\eta)\vee0}$ for any $t\in
[0,1]$. Fix $(t_j)_{0\leq j \leq q-1}\in[0,1]^q$ and assume $t_0\in
[\eta,1-\eta]$ (the other case can be proved similarly). Take a
non-decreasing measurable set $D\subset[0,1]^q$ and consider the set
$D'=\{ (M_{t_j})_{0\leq j \leq q-1}\in\mathbb{N}^q \telque
(G_{t_j}(M_{t_j}))_{0\leq j \leq q-1} \in D\},$ which is non-increasing
on $\mathbb{N}^q$ and measurable. We thus aim to prove that for any $n
\geq0$,
%
%
\renewcommand{\theequation}{\arabic{equation}}
\setcounter{equation}{30}
\begin{equation}
\label{PRDS-Poisson}\mbp\bigl[(M_{t_j})_{0\leq j \leq q-1}
\in D' | M_{t_0}=n+1\bigr] \leq\mbp
\bigl[(M_{t_j})_{0\leq j \leq q-1} \in D' |
M_{t_0}=n\bigr].
\end{equation}
Denote by $X_1 < \cdots< X_{k_X}$, $Y_1 < \cdots< Y_{k_Y}$ and $Z_1 <
\cdots< Z_{k_Z}$ the jump times of the process $(N_t)_{t\in[0,1]}$
within the (disjoint) subsets $[0,t_0-\eta)$, $[t_0-\eta,t_0+\eta]$
and $(t_0+\eta,1]$, respectively. Remark that $k_Y=M_{t_0}$ with our notation.
Since $(N_t)_{t\in[0,1]}$ is a Poisson process, the family $\{
(X_i,1\leq i \leq k_X, k_X),(Y_i,1\leq i \leq k_Y, k_Y),(Z_i,1\leq i
\leq k_Z, k_Z)\} $, contains mutually independent elements.
Furthermore, the distribution of $(Y_1 , \ldots,Y_{k_Y})$
conditionally on $k_Y=n$ is equal to the distribution of the order
statistics of a sample $(Y'_1,\ldots,Y'_n)$ of i.i.d. random variables
with common density $t\mapsto\lambda(t)/ \int_{ [t_0-\eta,t_0+\eta
]} \lambda(s) \,\mathrm{d}s $ on $[t_0-\eta,t_0+\eta]$ (w.r.t. the
Lebesgue measure).
Next, denoting $I_t=[(t-\eta)\vee0,(t+\eta)\wedge1]$, for any $t\in
[0,1]$, we can write:
\begin{eqnarray*}
&&\mbp\bigl[(M_{t_j})_{0\leq j \leq q-1} \in D'
| M_{t_0}=n+1\bigr]
\\
&&\quad=\mbp\Biggl[ \Biggl( \sum_{i=1}^{k_X}
{\mbf{1}\{X_i \in I_{t_j}\}}+\sum
_{i=1}^{n+1} {\mbf{1}\bigl\{Y'_i
\in I_{t_j}\bigr\}} +\sum_{i=1}^{k_Z}
{\mbf {1}\{Z_i \in I_{t_j}\}} \Biggr)_{0\leq j \leq q-1} \in
D' \Biggr]
\\
&&\quad=\mbp\Biggl[ \Biggl(\sum_{i=1}^{k_X}
{\mbf{1}\{X_i \in I_{t_j}\}}+\sum
_{i=1}^{n} {\mbf{1}\bigl\{Y'_i
\in I_{t_j}\bigr\}} +\sum_{i=1}^{k_Z}
{\mbf {1}\{Z_i \in I_{t_j}\}} \Biggr)_{0\leq j \leq q-1}\\
&&\hspace*{32pt}{} \in
D' - \bigl({\mbf{1}\bigl\{Y'_{n+1} \in
I_{t_j}\bigr\}} \bigr)_j \Biggr]
\\
&&\quad \leq\mbp\Biggl[ \Biggl(\sum_{i=1}^{k_X}
{\mbf{1}\{X_i \in I_{t_j}\}}+\sum
_{i=1}^{n} {\mbf{1}\bigl\{Y'_i
\in I_{t_j}\bigr\}} +\sum_{i=1}^{k_Z}
{\mbf {1}\{Z_i \in I_{t_j}\}} \Biggr)_{0\leq j \leq q-1} \in
D' \Biggr]
\\
&&\quad=\mbp\bigl[(M_{t_j})_{0\leq j \leq q-1} \in D'
| M_{t_0}=n\bigr],
\end{eqnarray*}
the inequality coming from $ D' - ({\mbf{1}\{Y'_{n+1} \in I_{t_j}\}})_j
\subset D'$, because $D'$ is non-increasing. This proves \eqref
{PRDS-Poisson} and concludes the proof.
\end{pf}
\end{appendix}

\section*{Acknowledgements}
This work was supported in part by the IST Programme of the European
Community, under the PASCAL Network of Excellence, IST-2002-506778.
The third author was supported by the French Agence Nationale de la Recherche
(ANR grant references: ANR-09-JCJC-0027-01, ANR-PARCIMONIE,
ANR-09-JCJC-0101-01).
The three authors where supported by the French ministry of foreign and european
affairs (EGIDE -- PROCOPE project number 21887 NJ).

\begin{supplement}
\stitle{Supplement to: ``Testing over a continuum of null hypotheses
with False Discovery Rate control''}
\slink[doi]{10.3150/12-BEJ488SUPP} 
\sdatatype{.pdf}
\sfilename{BEJ488\_supp.pdf}
\sdescription{This supplement provides some technical results and
introduces the so-called
general PRDS condition, which is a stronger assumption than the finite
dimensional PRDS condition.
This condition is useful to prove FDR control for procedures which are
not necessarily of the step-up type.}
\end{supplement}

%
%

\printhistory


\begin{thebibliography}{18}

\bibitem{BH1995}
%
\begin{barticle}[mr]
\bauthor{\bsnm{Benjamini},~\bfnm{Yoav}\binits{Y.}} \AND
\bauthor{\bsnm{Hochberg},~\bfnm{Yosef}\binits{Y.}}
(\byear{1995}).
\btitle{Controlling the false discovery rate: A practical and powerful approach
to multiple testing}.
\bjournal{J. Roy. Statist. Soc. Ser. B}
\bvolume{57}
\bpages{289--300}.
\bid{issn={0035-9246}, mr={1325392}}
\bptok{imsref}%
\end{barticle}
%
\endbibitem

\bibitem{BH1997}
%
\begin{barticle}[mr]
\bauthor{\bsnm{Benjamini},~\bfnm{Yoav}\binits{Y.}} \AND
\bauthor{\bsnm{Hochberg},~\bfnm{Yosef}\binits{Y.}}
(\byear{1997}).
\btitle{Multiple hypotheses testing with weights}.
\bjournal{Scand. J. Stat.}
\bvolume{24}
\bpages{407--418}.
\bid{doi={10.1111/1467-9469.00072}, issn={0303-6898}, mr={1481424}}
\bptok{imsref}%
\end{barticle}
%
\endbibitem

\bibitem{BY2001}
%
\begin{barticle}[mr]
\bauthor{\bsnm{Benjamini},~\bfnm{Yoav}\binits{Y.}} \AND
\bauthor{\bsnm{Yekutieli},~\bfnm{Daniel}\binits{D.}}
(\byear{2001}).
\btitle{The control of the false discovery rate in multiple testing under
dependency}.
\bjournal{Ann. Statist.}
\bvolume{29}
\bpages{1165--1188}.
\bid{doi={10.1214/aos/1013699998}, issn={0090-5364}, mr={1869245}}
\bptok{imsref}%
\end{barticle}
%
\endbibitem

\bibitem{Bill1999}
%
\begin{bbook}[mr]
\bauthor{\bsnm{Billingsley},~\bfnm{Patrick}\binits{P.}}
(\byear{1999}).
\btitle{Convergence of Probability Measures},
\bedition{2nd} ed.
\bseries{Wiley Series in Probability and Statistics: Probability and
Statistics}.
\blocation{New York}: \bpublisher{Wiley}.
\bid{doi={10.1002/9780470316962}, mr={1700749}}
\bptok{imsref}%
\end{bbook}
%
\endbibitem

\bibitem{BDR2011-supp}
%
\begin{bmisc}[auto:STB|2013/01/16|12:42:27]
\bauthor{\bsnm{Blanchard},~\bfnm{G.}\binits{G.}},
\bauthor{\bsnm{Delattre},~\bfnm{S.}\binits{S.}} \AND
\bauthor{\bsnm{Roquain},~\bfnm{E.}\binits{E.}}
(\byear{2014}).
\bhowpublished{Supplement to ``Testing over a continuum of null
hypotheses with
False Discovery Rate control.'' DOI:\doiurl{10.3150/12-BEJ488SUPP}}.
\bptok{imsref}%
\end{bmisc}
%
\endbibitem

\bibitem{BR2008}
%
\begin{barticle}[mr]
\bauthor{\bsnm{Blanchard},~\bfnm{Gilles}\binits{G.}} \AND
\bauthor{\bsnm{Roquain},~\bfnm{Etienne}\binits{E.}}
(\byear{2008}).
\btitle{Two simple sufficient conditions for {FDR} control}.
\bjournal{Electron. J. Stat.}
\bvolume{2}
\bpages{963--992}.
\bid{doi={10.1214/08-EJS180}, issn={1935-7524}, mr={2448601}}
\bptok{imsref}%
\end{barticle}
%
\endbibitem

\bibitem{CaiLow04}
%
\begin{barticle}[mr]
\bauthor{\bsnm{Cai},~\bfnm{T.~Tony}\binits{T.T.}} \AND
\bauthor{\bsnm{Low},~\bfnm{Mark~G.}\binits{M.G.}}
(\byear{2004}).
\btitle{An adaptation theory for nonparametric confidence intervals}.
\bjournal{Ann. Statist.}
\bvolume{32}
\bpages{1805--1840}.
\bid{doi={10.1214/009053604000000049}, issn={0090-5364}, mr={2102494}}
\bptok{imsref}%
\end{barticle}
%
\endbibitem

\bibitem{DS2001}
%
\begin{barticle}[mr]
\bauthor{\bsnm{D{\"u}mbgen},~\bfnm{Lutz}\binits{L.}} \AND
\bauthor{\bsnm{Spokoiny},~\bfnm{Vladimir~G.}\binits{V.G.}}
(\byear{2001}).
\btitle{Multiscale testing of qualitative hypotheses}.
\bjournal{Ann. Statist.}
\bvolume{29}
\bpages{124--152}.
\bid{doi={10.1214/aos/996986504}, issn={0090-5364}, mr={1833961}}
\bptok{imsref}%
\end{barticle}
%
\endbibitem

\bibitem{Ing82}
%
\begin{barticle}[mr]
\bauthor{\bsnm{Ingster},~\bfnm{Yu.~I.}\binits{Y.I.}}
(\byear{1982}).
\btitle{Minimax nonparametric detection of signals in white {G}aussian noise}.
\bjournal{Probl. Inf. Transm.}
\bvolume{18}
\bpages{130--140}.
\bptok{imsref}%
\end{barticle}
%
\endbibitem

\bibitem{Ing93}
%
\begin{barticle}[mr]
\bauthor{\bsnm{Ingster},~\bfnm{Yu.~I.}\binits{Y.I.}}
(\byear{1993}).
\btitle{Asymptotically minimax hypothesis testing for nonparametric
alternatives. I--{III}}.
\bjournal{Math. Methods Statist.}
\bvolume{2}
\bpages{85--114; 171--189; 249--268}.
\bptok{imsref}%
\end{barticle}
%
\endbibitem

\bibitem{PGVW2004}
%
\begin{barticle}[mr]
\bauthor{\bsnm{Perone~Pacifico},~\bfnm{M.}\binits{M.}},
\bauthor{\bsnm{Genovese},~\bfnm{C.}\binits{C.}},
\bauthor{\bsnm{Verdinelli},~\bfnm{I.}\binits{I.}} \AND
\bauthor{\bsnm{Wasserman},~\bfnm{L.}\binits{L.}}
(\byear{2004}).
\btitle{False discovery control for random fields}.
\bjournal{J. Amer. Statist. Assoc.}
\bvolume{99}
\bpages{1002--1014}.
\bid{doi={10.1198/0162145000001655}, issn={0162-1459}, mr={2109490}}
\bptok{imsref}%
\end{barticle}
%
\endbibitem

\bibitem{PGVW2007}
%
\begin{barticle}[mr]
\bauthor{\bsnm{Perone~Pacifico},~\bfnm{M.}\binits{M.}},
\bauthor{\bsnm{Genovese},~\bfnm{C.}\binits{C.}},
\bauthor{\bsnm{Verdinelli},~\bfnm{I.}\binits{I.}} \AND
\bauthor{\bsnm{Wasserman},~\bfnm{L.}\binits{L.}}
(\byear{2007}).
\btitle{Scan clustering: A false discovery approach}.
\bjournal{J. Multivariate Anal.}
\bvolume{98}
\bpages{1441--1469}.
\bid{doi={10.1016/j.jmva.2006.11.011}, issn={0047-259X}, mr={2364129}}
\bptok{imsref}%
\end{barticle}
%
\endbibitem

\bibitem{Rob2002}
%
\begin{barticle}[mr]
\bauthor{\bsnm{Robin},~\bfnm{St{\'e}phane}\binits{S.}}
(\byear{2002}).
\btitle{A compound {P}oisson model for word occurrences in {DNA} sequences}.
\bjournal{J.~Roy. Statist. Soc. Ser. C}
\bvolume{51}
\bpages{437--451}.
\bid{doi={10.1111/1467-9876.00279}, issn={0035-9254}, mr={1974071}}
\bptok{imsref}%
\end{barticle}
%
\endbibitem

\bibitem{SMD2011}
%
\begin{bmisc}[auto:STB|2013/01/16|12:42:27]
\bauthor{\bsnm{Schmidt-Hieber},~\bfnm{J.}\binits{J.}},
\bauthor{\bsnm{Munk},~\bfnm{A.}\binits{A.}} \AND
\bauthor{\bsnm{Duembgen},~\bfnm{L.}\binits{L.}}
(\byear{2011}).
\bhowpublished{Multiscale methods for shape constraints in deconvolution:
Confidence statements for qualitative features. Available at
arXiv:\arxivurl{1107.1404}.}
\bptok{imsref}%
\end{bmisc}
%
\endbibitem

\bibitem{SGA2011}
%
\begin{barticle}[auto:STB|2013/01/16|12:42:27]
\bauthor{\bsnm{Schwartzman},~\bfnm{A.}\binits{A.}},
\bauthor{\bsnm{Gavrilov},~\bfnm{Y.}\binits{Y.}} \AND
\bauthor{\bsnm{Adler},~\bfnm{R.~J.}\binits{R.J.}}
(\byear{2011}).
\btitle{Multiple testing of local maxima for detection of peaks in 1{D}}.
\bjournal{Ann. Statist.}
\bvolume{39}
\bpages{3290--3319}.
\bptok{imsref}%
\end{barticle}
%
\endbibitem

\bibitem{See1968}
%
\begin{barticle}[auto:STB|2013/01/16|12:42:27]
\bauthor{\bsnm{Seeger},~\bfnm{P.}\binits{P.}}
(\byear{1968}).
\btitle{A note on a method for the analysis of significances en masse}.
\bjournal{Technometrics}
\bvolume{10}
\bpages{586--593}.
\bptok{imsref}%
\end{barticle}
%
\endbibitem

\bibitem{vdLDudPol06}
%
\begin{barticle}[mr]
\bauthor{\bparticle{van~der} \bsnm{Laan},~\bfnm{Mark~J.}\binits{M.J.}},
\bauthor{\bsnm{Dudoit},~\bfnm{Sandrine}\binits{S.}} \AND
\bauthor{\bsnm{Pollard},~\bfnm{Katherine~S.}\binits{K.S.}}
(\byear{2004}).
\btitle{Augmentation procedures for control of the generalized family-wise
error rate and tail probabilities for the proportion of false positives}.
\bjournal{Stat. Appl. Genet. Mol. Biol.}
\bvolume{3}
\bpages{27 pp. (electronic)}.
\bid{doi={10.2202/1544-6115.1042}, issn={1544-6115}, mr={2101464}}
\bptnote{check year}%
\bptok{imsref}%
\end{barticle}
%
\endbibitem

\end{thebibliography}
\end{document}